\documentclass[11pt,letterpaper]{article}
\usepackage[dvipsnames]{xcolor}
\usepackage{tikz}
\usetikzlibrary{angles,quotes,calc}
\usetikzlibrary{
  shapes.multipart,
  matrix,
  positioning,
  shapes.callouts,
  shapes.arrows,
  calc}

\usepackage{answers}
\usepackage{setspace}

\usepackage{graphicx}
\usepackage[shortlabels]{enumitem}
\usepackage{multicol}
\usepackage{mathrsfs}
\usepackage[margin=1in]{geometry} 
\usepackage{amsmath,amsthm,amssymb}
\usepackage{mathtools}
\usepackage[utf8]{inputenc}
\usepackage[english]{babel}

\usepackage[maxalphanames=99, maxbibnames=99,style=alphabetic,natbib=true]{biblatex}
\usepackage{bbm}
\usepackage{breqn}

\usepackage{subfiles}
\usepackage{amsmath,amssymb}
\usepackage[english]{babel}
\usepackage[nottoc]{tocbibind}
\usepackage{bm}
\usepackage{hyperref}
\usepackage{algorithm}
\usepackage{multirow}
\usepackage{multicol}
\usepackage{algpseudocode}
\usepackage[capitalise, nameinlink, noabbrev]{cleveref}
\usepackage{xifthen}
\usepackage{booktabs}
\hypersetup{
    colorlinks=true,
    linkcolor=blue,
    filecolor=magenta,      
    urlcolor=cyan,
    citecolor=Green
}
\usepackage{thmtools}
\usepackage{thm-restate}

\usepackage{gensymb}

\newtheorem{problem}{Problem}

\newtheorem{theorem}{Theorem}

\newtheorem{lemma}[theorem]{Lemma}

\newtheorem{definition}[theorem]{Definition}
\newtheorem{remark}[theorem]{Remark}

\newtheorem{proposition}[theorem]{Proposition}

\newcommand{\Z}{\mathbb{Z}}
\newcommand{\C}{\mathbb{C}}
\newcommand{\R}{\mathbb{R}}

\DeclareMathOperator*{\E}{\mathbb{E}}

\DeclareMathOperator*{\poly}{poly}
\DeclareMathOperator*{\supp}{supp}
\DeclareMathOperator*{\Tr}{Tr}

\newcommand{\eps}{\varepsilon}

\newcommand{\norm}[1]{\left\lVert #1\right\rVert}

\newcommand{\abs}[1]{\left|#1\right| }

\newcommand{\wt}{\widetilde}

\newcommand{\tuple}[1]{\left(#1\right)}

\newcommand{\zo}{\{0, 1\}}
\newcommand{\Ext}{\textup{\textsf{Ext}}}
\newcommand{\Smp}{\textup{\textsf{Samp}}}

\newcommand{\Prb}[2][]{ \ifthenelse{\isempty{#1}}
  {\Pr\left[#2\right]}
  {\Pr_{#1}\left[#2\right]} }
\newcommand{\Ex}[2][]{ \ifthenelse{\isempty{#1}}
  {\E\left[#2\right]}
  {\E_{#1}\left[#2\right]} }
\newcommand{\Vari}[2][]{ \ifthenelse{\isempty{#1}}
  {\mathbf{Var}\left[#2\right]}
  {\mathop{\mathbf{Var}}_{#1}\left[#2\right]} }
\newcommand{\nPrb}[2][]{ \ifthenelse{\isempty{#1}}
  {\Pr[#2]}
  {\Pr_{#1}[#2]} }
\newcommand{\nEx}[2][]{ \ifthenelse{\isempty{#1}}
  {\E[#2]}
  {\E_{#1}[#2\right]} }

\author{Zhiyang Xun\thanks{Supported by NSF Grant CCF-2312573, a Simons Investigator Award (\#409864, David Zuckerman), NSF award CCF-2008868
and the NSF AI Institute for Foundations of Machine Learning (IFML).}\\Department of Computer Science\\The University of Texas at Austin\\ \texttt{zxun@cs.utexas.edu} \and David Zuckerman\thanks{Supported by NSF Grant CCF-2312573 and a Simons Investigator Award (\#409864).}\\Department of Computer Science\\The University of Texas at Austin\\ \texttt{diz@utexas.edu}}

\title{Near-Optimal Averaging Samplers and Matrix Samplers}

\begin{document}
\maketitle
\begin{abstract}
We present the first efficient averaging sampler that achieves asymptotically optimal randomness complexity and near-optimal sample complexity. For any $\delta < \varepsilon$ and any constant $\alpha > 0$, our sampler uses $m + O(\log (1 / \delta))$ random bits to output $t = O((\frac{1}{\varepsilon^2} \log \frac{1}{\delta})^{1 + \alpha})$ samples $Z_1, \dots, Z_t \in \{0, 1\}^m$ such that for any function $f: \{0, 1\}^m \to [0, 1]$,
\[
\Pr\left[\left|\frac{1}{t}\sum_{i=1}^t f(Z_i) - \mathbb{E}[f]\right| \leq \varepsilon\right] \geq 1 - \delta.
\]
The randomness complexity is optimal up to a constant factor, and the sample complexity is optimal up to the $O((\frac{1}{\varepsilon^2} \log \frac{1}{\delta})^{\alpha})$ factor.

Our technique generalizes to matrix samplers. A matrix sampler is defined similarly, except that $f: \{0, 1\}^m \to \mathbb{C}^{d \times d}$ and the absolute value is replaced by the spectral norm.
Our matrix sampler achieves randomness complexity $m + \widetilde O (\log(d / \delta))$ and sample complexity $ O((\frac{1}{\varepsilon^2} \log \frac{d}{\delta})^{1 + \alpha})$  for any constant $\alpha > 0$, both near-optimal with only a logarithmic factor in randomness complexity and an additional $\alpha$ exponent on the sample complexity.

We use known connections with randomness extractors and list-decodable codes to give applications to these objects. Specifically, we give the first extractor construction with optimal seed length up to an arbitrarily small constant factor above 1, when the min-entropy $k = \beta n$ for a large enough constant $\beta < 1$. Finally, we generalize the definition of averaging sampler to any normed vector space.
\end{abstract}

\section{Introduction}
Randomization plays a crucial role in computer science, offering significant benefits across various applications. However, obtaining true randomness can be challenging.
It's therefore natural to study whether we can achieve the benefits of randomization while using few random bits.

One of the most basic uses of randomness is sampling. 
Given oracle access to an arbitrary function $f: \zo^m \to [0, 1]$ on a large domain, our goal is to estimate its average value.
By drawing $t = O(\log(1/\delta)/\eps^2)$ independent random samples $Z_1, \dots, Z_t \in \zo^m$, 
the Chernoff bound guarantees that the average value $\abs{\frac{1}{t} \sum_{i=1}^t f(Z_i) - \E f} \le \eps$ with probability at least $1 - \delta$. This method uses full independence in sampling, but more efficient strategies can be pursued. This leads to the following definition:
\begin{definition}[\cite{bellare1994randomness}]
    A function $\Smp : \zo^n \to (\zo^m)^t$ is a $(\delta, \eps)$ averaging sampler with $t$ samples using $n$ random bits if for every function $f : \zo^m \to [0, 1]$, we have \[
    \Pr_{(Z_1, \dots, Z_t) \sim \Smp(U_n)}
    \left[ \abs{\frac{1}{t}  \sum_i f(Z_i) - \E f} \leq \varepsilon \right]
    \geq 1 - \delta.
    \]
\end{definition}

The goal is to construct explicit samplers using a small number of random bits that have sample complexity close to the optimal. Researchers have made significant progress toward this goal, and a summary is presented in \cref{table:samplers}. 
Bellare and Rompel \cite{bellare1994randomness} suggested that interesting choices of parameters are $\delta = \exp(-\poly(m))$ and $\eps = 1/\poly(m)$, which allow us to use $\poly(m)$ random bits and generate $\poly(m)$ samples.
For simplicity, we assume $\delta \le \eps$ throughout the paper (see \cref{rem:choice_of_parameters} for further discussion).

\begin{table}[ht]
\centering
\makebox[\textwidth][c]{
\begin{tabular}{|c|c|c|c|c|}
\hline
Reference & Method & Random Bits & Sample Complexity \\
\hline
\hline
\cite{CEG95} & Lower Bound & $ m + \log (1 / \delta) - \log(O(t)) $ & $ \Omega( \log (1 / \delta) / \eps^2) $ \\
\hline
\cite{CEG95} & Non-Explicit & $ m + 2\log (2 / \delta) + \log\log(1 / \eps) $ & $ 2 \log (4 / \delta) / \eps^{2} $ \\
\hline
\hline
Standard & Full Independence & $ O(m \log (1 / \delta) / \eps^2) $ & $ O( \log (1 / \delta) / \eps^{2}) $ \\
\hline
\cite{CG89} & Pairwise Independence & $O(m + \log(1 / \delta))$  & $ O(1 / (\delta\eps^2))$ \\
\hline
\cite{Gil93} & Expander Walks &  $ m + O( \log (1 / \delta) / \eps^{2}) $ & $ O( \log (1 / \delta) / \eps^{2}) $ \\
\hline
\cite{bellare1994randomness} & Iterated Sampling &  $ O(m + (\log m) \log (1 / \delta)) $ & $ \text{poly}(1 / \eps, \log (1 / \delta), \log m) $ \\
\hline
\cite{zuckerman1997randomness} & Hash-Based Extractors & $ (1 + \alpha)(m + \log (1 / \delta)) $ & $ \text{poly}(1 / \eps, \log (1 / \delta), m) $ \\
\hline
\cite{reingold2000entropy} & Zig-Zag Extractors & $ m + (1 + \alpha)\log (1 / \delta) $ & $ \text{poly}(1/\eps, \log (1 / \delta)) $ \\
\hline

\hline
\multirow{ 2}{*}{\cref{cor:domain_sampler}} & Compose \cite{reingold2000entropy} &  \multirow{ 2}{*}{$m + O(\log (1 / \delta))$}  & \multirow{ 2}{*}{$O( ( \log (1 / \delta) / \eps^2)^{1 + \alpha}  )$} \\
 &With Almost $\ell$-wise Ind. &  & \\
\hline
\end{tabular}
}
\caption{Comparison of averaging samplers, $\alpha$ any positive constant.}
\label{table:samplers}
\end{table}

The best existing randomness-efficient averaging sampler comes from the equivalence between averaging samplers and extractors \cite{zuckerman1997randomness}, which we will elaborate on later in the paper. Improving Zuckerman's construction, Reingold, Vadhan, and Wigderson \cite{reingold2000entropy} provided a $(\delta, \eps)$ averaging sampler for domain $\zo^{m}$ that uses $m + (1 + \alpha)\log (1 / \delta)$ random bits for any positive constant $\alpha$. This almost matches the lower bound in~\cite{CEG95}. However, a notable gap remains in sample complexity: the existing construction's complexity $\poly(1 / \eps, \log (1 / \delta))$ does not align with the optimal $O( \log (1 / \delta)  / \eps^2)$. This raised the following open problem (see, e.g., \cite[Open Problem 4.24]{vadhan2012pseudorandomness}, \cite[Section 6]{oded_survey}):
\begin{problem}
  \label{problem1}
  {Can we explicitly design a $(\delta, \eps)$ averaging sampler for domain $\zo^m$ that uses $O(m + \log (1 / \delta))$ random bits and only $O( \log (1 / \delta)  / \eps^2)$ samples?}
\end{problem}

We note that such algorithms do exist for general samplers, which query $f$ and estimate $\E f$ through a more complicated  computation than taking the average~\cite{BGG93}.  
However, many applications require the use of averaging samplers, such as the original use in interactive proofs \cite{bellare1994randomness}. 
Beyond these applications, averaging samplers act as a fundamental combinatorial object that relate to other notions such as randomness extractors, expander graphs, and list-decodable codes \cite{zuckerman1997randomness, vadhan2007unified}.

\subsection{Our Averaging Sampler} 
In this paper, we construct a polynomial-time computable $(\delta, \eps)$ averaging sampler with near-optimal sample complexity using an asymptotically optimal number of random bits.  In fact, the sampler we constructed is a \emph{strong} sampler, defined as follows:

\begin{definition}
    A $(\delta, \varepsilon)$ averaging sampler $\Smp$ is \emph{strong} if
    for every sequence of $t$ functions $f_1, \dots, f_t: \zo^m \to [0,1]$, we have
    \[
    \Pr_{(Z_1, \dots, Z_t) \sim \Smp(U_n)}
    \left[ \abs{\frac{1}{t}  \sum_i \tuple{f_i(Z_i) - \E f_i}} \leq \varepsilon \right]
    \geq 1 - \delta.
    \]
\end{definition}

We now state our main theorems about averaging samplers, which follow from a more general theorem that is slightly more complicated to state,~\cref{thm:general_sampler}.

\begin{restatable}{mtheorem}{CorollaryTwo}
    \label{cor:domain_sampler}
     For every constant $\alpha > 0$, there exists an efficient strong $(\delta, \eps)$ averaging sampler for domain $\zo^m$ that uses  $m + O(\log(1/\delta))$ random bits and  $O((\frac{1}{\eps^{2}} \log \frac{1}{\delta})^{1 + \alpha})$ samples.
\end{restatable}

This nearly resolves~\cref{problem1}. We also give a sampler with asymptotically optimal sample complexity but a worse randomness complexity.


\begin{restatable}{mtheorem}{optimalSampleSampler}
    \label{lem:optimal_sample_sampler}
    There exists an efficient strong $(\delta, \eps)$ averaging sampler for domain $\zo^m$ that uses $m + O(\log\frac{1}{\delta} (\log \frac{1}{\eps} + \log\log\frac{1}{\delta} ))$ random bits and $O(\frac{1}{\eps^{2}} \log \frac{1}{\delta})$ samples. 
\end{restatable}

\subsection{Matrix Samplers}
A natural generalization of the classic Chernoff bound is the Matrix Chernoff Bound~\cite{Rud99, AW02, Tro12}. Suppose we wish to estimate $\E f$ for a matrix-valued function $f: \zo^m \to \C^{d \times d}$ satisfying $\norm{f(x)} \le 1$.
By drawing $t = O(\log(d/\delta)/\eps^2)$ independent random samples $Z_1, \dots, Z_t \in \zo^m$, the
Matrix Chernoff Bound guarantees that \[
    \Prb{\norm{\frac{1}{t} \sum_{i=1}^t f(Z_i) - \E f} \le \eps} \ge 1 - \delta,
\]
where $\norm{\cdot}$ denotes the spectral norm.
As in the real-valued case, we wish to derandomize this process without increasing the sample complexity too much. 
To address this, Wigderson and Xiao~\cite{WX05} initiated the study of randomness-efficient matrix samplers:
\begin{definition}
    A function $\Smp : \zo^n \to (\zo^m)^t$ is a $d$-dimensional $(\delta, \eps)$ matrix sampler with $t$ samples using $n$ random bits if the following holds: For any function $f : \zo^m \to \C^{d \times d}$ such that $\|f(x)\| \leq 1$ for all $x \in \zo^m$, we have
    \[
    \Pr_{(Z_1, \dots, Z_t) \sim \Smp(U_n)}
    \left[ \norm{\frac{1}{t}  \sum_i f(Z_i) - \E f} \leq \varepsilon \right]
    \geq 1 - \delta.
    \]
\end{definition}


Extending the construction of non-explicit standard averaging samplers~\cite{CEG95}, we can show that there exists a non-explicit matrix sampler that requires only an additional $2 \log d$ bits of randomness compared to averaging samplers while achieving asymptotically optimal sample complexity.
\begin{restatable}{proposition}{nonExplicit}
    \label{prop:non_explicit}
    There exists a (non-explicit) $d$-dimensional $(\delta, \eps)$ matrix sampler for domain $\zo^m$ using $O(\frac{1}{\eps^2}\log\frac{d}{\delta})$ samples and $m + 2\log\frac{1}{\delta} + 2 \log d + \log \log \frac{d}{\eps}$ random bits.
\end{restatable}
However, explicitly constructing randomness-efficient matrix samplers turns out to be very challenging. While a union bound over matrix entries suggests that a randomness-optimal averaging sampler directly implies a randomness-optimal matrix sampler (see~\cref{lem:sampler_implies_matrix_sampler}), this method incurs an unavoidable $d^2$ factor in sample complexity, making the dependence on $d$ exponentially worse than optimal.
This raises an open question: can we construct a matrix sampler with (nearly) optimal randomness complexity and polynomial sample complexity, analogous to the averaging samplers in~\cite{bellare1994randomness} and~\cite{zuckerman1997randomness}? 

\begin{problem}
    \label{problem2}
    Can we explicitly design a $d$-dimensional $(\delta, \eps)$ matrix sampler for domain $\zo^m$ that uses $O(m + \log(d / \delta))$ random bits and $\poly(1/\eps, \log(1 / \delta), \log d)$ samples?
\end{problem}

\begin{table}[t]
\centering
\makebox[\textwidth][c]{
\begin{tabular}{|c|c|c|c|c|}
\hline
Reference & Method & Random Bits & Sample Complexity \\
\hline
\hline
\cref{prop:non_explicit} & Non-Explicit & $ m + 2\log (1 / \delta) + 2 \log d $ & $ O (\log (d / \delta) / \eps^{2}) $ \\
\hline
\hline
\cite{AW02} & Matrix Chernoff Bound & $ O(m \log (d / \delta) / \eps^2) $ & $ O( \log (d / \delta) / \eps^{2}) $ \\
\hline
\cite{WX05} & Union Bound Over Entries & $m + O(\log(d / \delta))$  & $O( (d / \eps)^{2 + \alpha} \cdot \log^{1 + \alpha} (1 / \delta))$ \\
\hline
\cite{GLSS18} & Expander Walks &  $ m +O( (1 / \eps^2) \cdot \log (d / \delta)) $ & $ O(  \log (d / \delta)  / \eps^2 ) $ \\
\hline
\hline
\cref{thm:sampler_with_replacement} &  Iterated Sampler Composition &  $ m + O( \log (1 / \delta) + \log d \log\log d) $ & $ O((\log (d / \delta) / \eps^{2})^{1 + \alpha}) $ \\
\hline
\end{tabular}
}
\caption{Comparison of matrix samplers, $\alpha$ any positive constant, $\eps = 1 / \poly(m)$, $\delta = \exp(-\poly(m))$, ignoring lower order terms. The complexity of the union bound sampler depends on the complexity of the ``base'' averaging sampler, and we use the bound in \cref{cor:domain_sampler} here.}
\label{table:matrix_samplers}
\end{table}

We summarize prior matrix sampler constructions  in~\cref{table:matrix_samplers}.
The best existing construction, a matrix analog of the expander walks sampler, was provided by Garg, Lee, Song, and Srivastava~\cite{GLSS18}. 
Similar to expander walks for real-valued sampling, this construction gives asymptotically optimal sample complexity, but the randomness complexity is worse than optimal by a $\poly(1 / \eps)$ factor. We note that even if we allow the the matrix sampler to be non-averaging, no known better construction is currently known.

In this work, we construct a polynomial-time computable $(\delta, \eps)$ matrix sampler with near-optimal randomness and sample complexity. The randomness complexity is optimal up to a logarithmic factor, and the sample complexity is within a $(\frac{1}{\eps^2} \log \frac{d}{\delta})^\alpha$ factor of optimal for arbitrarily small constant $\alpha > 0$. This brings us close to resolving~\cref{problem2}.

\begin{restatable}{mtheorem}{samplerWithReplacement}
 \label{thm:sampler_with_replacement}
    For any constant $\alpha > 0$:
    There exists an efficient $d$-dimensional $(\delta, \eps)$ matrix sampler for domain $\zo^m$ that uses $m + O(\log(1 / \delta) + \log(d / \eps) \log\log d)$ random bits and $O((\frac{1}{\eps^2}\log\frac{d}{\delta})^{1 + \alpha})$ samples.   
\end{restatable}

Additionally, we construct a matrix sampler achieving asymptotically optimal randomness complexity, though at the cost of increased sample complexity. This breaks the $d^2$ barrier in sample complexity for randomness-optimal matrix samplers.

\begin{restatable}{mtheorem}{polydSampler}
    \label{thm:polydSampler}
        For any constant $\alpha > 0$,  there exists an efficient $d$-dimensional $(\delta, \eps)$ matrix sampler for domain $\zo^m$ that uses $m + O(\log (d / \delta))$ random bits and $O(\frac{d^\alpha}{\eps^{2+\alpha}}\log^{1+\alpha}\frac{1}{\delta})$ samples.
\end{restatable}

\subsection{Samplers for General Normed Vector Space}
Apart from spectral norms of matrices, it is natural to study the averaging-sampling problem in other normed vector spaces \(V\).

\begin{definition}
    A function \(\Smp : \zo^n \to (\zo^m)^t\) is a \((V,\mathcal{F})\)-sampler for a normed space \(V\) and a class of functions \(\mathcal{F} \subseteq \{f:\zo^m \to V\}\) if, for every \(f \in \mathcal{F}\),
    \[
    \Pr_{(Z_1, \dots, Z_t) \sim \Smp(U_n)}
    \left[ \norm{\frac{1}{t}  \sum_i f(Z_i) - \E f} \leq \varepsilon \right]
    \geq 1 - \delta.
    \]
    We call \(\Smp\) a \(V\)-sampler when \(\mathcal{F} = \{f:\zo^m \to V \mid \|f(x)\|\le 1 \text{ for all } x\}\).
\end{definition}


Under this definition, a \(d\)-dimensional matrix sampler is precisely a \((\C^{d\times d},\|\cdot\|_2)\)-sampler. Previous work also studied $(\R, \mathcal{F})$-samplers for a broader class of $\mathcal{F}$, such as subgaussian or subexponential real-valued functions \cite{Bla19,Agr19}. Extending our construction to other normed spaces and broader function classes remains an interesting direction for future research.

 \subsection{Randomness Extractors}

Our sampler construction has implications for randomness extractors.
A randomness extractor is a function that extracts almost-uniform bits from a low-quality source of randomness. We define the quality of a random source as its min-entropy.
\begin{definition}
    The min-entropy of a random variable $X$ is 
    
    \[ H_{\infty}(X) := \min_{x \in \supp(X)} \log\left(\frac{1}{\Pr[X = x]}\right). \]

An $(n,k)$-source is a random variable on $n$ bits with min-entropy at least $k$.
\end{definition}

Then a randomness extractor is defined as:

\begin{definition}[\cite{Nisan_Zuckerman}]
A function $\Ext : \{0,1\}^n \times \{0,1\}^d \rightarrow \{0,1\}^m$ is a $(k,\varepsilon)$ extractor if for every $(n,k)$-source $X$, the distribution $\Ext(X,U_d) \approx_{\varepsilon} U_m$. We say $\Ext$ is a \emph{strong} $(k,\varepsilon)$ extractor if for every $(n,k)$-source $X$, the distribution $(\Ext(X,Y), Y) \approx_{\varepsilon} U_{m+d}$, where $Y$ is chosen from $U_d$.
\end{definition}
 
Randomness extractors are essential tools in theoretical computer science. However, there has been little study of explicit extractors with the right dependence on $\eps$ for vanishing $\eps$. 
This is a particular concern in cryptography, where extractors are widely used as building blocks and security requirements demand superpolynomially small $\eps$~\cite{Lu2002, Vadhan2003, Canetti2000, Dodis2002, Kalai2008, Kalai2009, Dodis2009}. 
Existentially, there are extractors with seed length $d = \log(n - k) + 2 \log(1 / \eps) + O(1)$, and there is a matching lower bound \cite{RT00}. 

Zuckerman~\cite{zuckerman1997randomness} showed that averaging samplers are essentially equivalent to extractors. Specifically, an extractor $\Ext: \zo^n \times [2^d] \to \zo^m$ can be seen as a sampler that generates $\Ext(X, i)$ as its $i$-th sample point using the random source $X$. Using this equivalence, we give the first extractor construction with optimal seed length up to an arbitrarily small constant factor bigger than 1, when the min-entropy $k = \beta n$ for a large enough constant $\beta < 1$.

\begin{restatable}{mtheorem}{ourExtractor}
    \label{thm:our_extractor}
    For every constant $\alpha > 0$, there exists constant $\beta < 1$ such that for all $\eps > 0$ and $k \ge  \beta n$, there is an efficient strong $(k, \eps)$ extractor $\Ext : \zo^n \times \zo^d \to \zo^m$ with $m = \Omega(k) - \log(1 / \eps)$ and $d = (1 + \alpha)\log (n-k) + (2 + \alpha) \log (1 / \eps) + O(1)$.
\end{restatable}

Prior to our work, extractors with a seed length dependence on $\eps$ achieving $2\log(1/\eps)$ or close to it were based on the leftover hash lemma~\cite{ILL89, BBR, IZ, HILL} and expander random walks~\cite{Gil93, Zuc07}. Extractors using the leftover hash lemma have a seed length of $n + 2 \log(1 / \eps)$, which is far from optimal. Expander random walks give a $(k, \eps)$ extractor with $k > (1 - \Omega(\eps^2))n$ and an optimal seed length of $\log(n - k) + 2\log(1 / \eps) + O(1)$. Our extractor is better than expander walks for all vanishing $\eps$ by allowing smaller entropy $k$.

In fact, if we aim to remove the $\alpha$ and achieve the optimal seed length of $\log(n - k) + 2\log(1 / \eps) + O(1)$ to match expander random walks, we can set $s = 1$ in \cref{thm:general_sampler} and get the following extractor for entropy rate $1-O(1/\log n)$ for $\eps \ge 1/\poly(n)$:

\begin{restatable}{mtheorem}{OurOptimalExtractor}
    \label{thm:our_optimal_extractor}
    There exists constant $\beta < 1$ such that for all $\eps > 0$ and $k \ge ( 1 - \frac{\beta}{\log n + \log(1/\eps)} ) n$, there is an efficient strong $(k, \eps)$ extractor $\Ext : \zo^n \times \zo^d \to \zo^m$ with $m = \Omega(k) - \log^2(1 / \eps)$ and $d = \log (n-k) + 2\log (1 / \eps) + O(1)$.
\end{restatable}

This is better than expander random walks' entropy rate of $1 - O(\eps^2)$ for all $\eps \le o(1 / \sqrt{\log n})$.

\subsection{List-Decodable Codes}

Another perspective on averaging samplers is its connection to error-correcting codes. 
Ta-Shma and Zuckerman~\cite{ExtractorCodes} showed that strong randomness extractors are equivalent to codes with good soft-decision decoding, which is related to list recovery.
From this perspective, the composition scheme in our construction is similar to code concatenation. 

For codes over the binary alphabet, soft decision decoding amounts to list decodability, which we focus on here.

We give good list-decodable codes without using the composition. That is, by just applying our almost $\ell$-wise independence sampler on the binary alphabet, we can get a binary list-decodable code with rate $\Omega(\eps^{2 + \alpha})$ and non-trivial list size, although the list size is still exponential.

\begin{restatable}{mtheorem}{OurCode}
    \label{thm:our_code}
    For every constant $\alpha > 0$: there exists an explicit binary code with rate $\Omega(\eps^{2 + \alpha})$ that is $(\rho = \frac{1}{2} - \eps, L)$ list-decodable with list size $L = 2^{(1 - c)n}$ for some constant $c = c(\alpha) > 0$.    
\end{restatable}

Prior to our work, the best known code rate was $\Omega(\eps^3)$ by Guruswami and Rudra~\cite{GR08}. We emphasize that their code achieved a list size of $L = \poly(n)$, while our list size is exponentially large, making our code unlikely to be useful.

\subsection{Techniques}
\subsubsection{Averaging Samplers}
Our construction of the averaging sampler is very simple, and is based on two observations:

\begin{enumerate}
    \item Rather than querying every sample point produced by a sampler $\Smp$, we can use an inner sampler $\Smp_{in}$ to select a subset of samples for querying.  This sub-sampling approach has been utilized in previous sampler constructions~\cite{bellare1994randomness, oded_survey}. Although $\Smp_{in}$ incurs an additional randomness cost, the final sample complexity depends only on $\Smp_{in}$, leading to reduced overall sample complexity. Since the domain of $\Smp_{in}$ is much smaller than the original domain, we can leverage more efficient sampling strategies.
    
\item  The bottleneck of generating an almost $\ell$-wise independent  sequence over a large domain $\zo^m$ lies in sampling $\ell$ independent random points, which costs $\ell m$ random bits. Since we can only afford $O(m)$ random bits, we are restricted to generating constant-wise independent  samples. However, for a much smaller domain, we can use few random bits to generate an almost $\ell$-wise independent  sequence for large~$\ell$. \label{o2}
\end{enumerate}

Our construction is outlined as follows. 
 Let $\Smp_{E}: \zo^n \times [t'] \to \zo^m$ be the extractor-based sampler in \cite{reingold2000entropy}.
 Let $Y_1, \dots, Y_{t}$ be an almost $\ell$-wise independent  sequence over domain $[t']$, thinking of $t \ll t'$. Our sampler is then defined by \[
    \Smp := (\Smp_{E}(X, Y_1), \Smp_{E}(X, Y_2), \dots, \Smp_{E}(X, Y_{t})).
 \]
In this construction, we use the almost $\ell$-wise independent  sequence to sub-sample from the extractor-based sampler. This can be viewed as a composition, similar to other cases such as Justesen codes~\cite{Jus72} and the first PCP theorem~\cite{PCP}, where the goal is to optimize two main parameters simultaneously by combining two simpler schemes, each optimizing one parameter without significantly compromising the other.

Previous works have applied almost $\ell$-wise independence in extractor constructions. Srinivasan and Zuckerman~\cite{SZ99} proved a randomness-efficient leftover hash lemma by sampling an almost $\ell$-wise independent function $f$ using uniform seeds $U_d$ and output $f(X)$, where $X$ is the weak random source. From an extractor perspective, our inner sampler takes an inverse approach: we use $X$ to pick a function $f$ in the space of almost $\ell$-wise independent functions, and then output $f(U_d)$. Furthermore, Raz's two-source extractor~\cite{Raz05} follows a more general framework, where two weak random sources to sample are used -- one to sample an almost $\ell$-wise independent function and the other as its input. However, directly applying Raz's error bound in our analysis~(\cref{lem:general_ell_wise_sampler}) results in a sample complexity that is off by a $\log (1 / \delta)$ factor.

\subsubsection{Matrix Samplers}
Using the connection between averaging samplers and matrix samplers (see \cref{lem:sampler_implies_matrix_sampler}), our averaging sampler directly implies a $(\delta, \eps)$ matrix sampler using $m + O(\log(d / \delta))$ random bits and $O((\frac{d^2}{\eps^2} \log \frac{1}{\delta})^{1 + \alpha})$ samples. This already gives the best randomness-optimal matrix sampler to date; however, its sample complexity has exponentially worse dependence on $d$ than optimal. 

Our sub-sampling technique using almost $\ell$-wise independence offers a way to further reduce sample complexity.
The composition of samplers only depends on the triangle inequality, which also applies to spectral norms.
The remaining task is to verify that almost $\ell$-wise independence also provides good concentration bounds for matrix sampling, which is straightforward given the extensive literature on moment inequalities for random matrices~\cite{CGT12, LT13, Tro15}.

Applying this composition, we get a $(\delta, \eps)$ matrix sampler using $m + O(\log(d / \delta))$ random bits and $O((\frac{d^\alpha}{\eps^{2 + \alpha}} \log^{1 + \alpha} \frac{1}{\delta}))$ samples, as described in \cref{thm:polydSampler}. This is close to optimal for cases where $d < \poly(1 / \eps, \log(1 / \delta))$, though it is not yet sufficient for 
larger $d$.

However, we can apply composition recursively. By repeating the composition $O(\log\log d)$ times, the dependence on $d$ becomes $d^{\alpha^{O(\log\log d)}} = O(1)$. Each round of composition costs an additional $O(\log (d / \delta))$ random bits, resulting in a $(\delta, \eps)$ matrix sampler using $m + O(\log(d /\delta) \log \log d)$ random bits and $O((\frac{1}{\eps^2}\log\frac{d}{\delta})^{1 + \alpha})$ samples. This already gives a  matrix sampler using $m + \wt{O}(\log(d / \delta))$ random bits and near-optimal sample complexity.

To further improve the dependence on $\delta$ in randomness complexity and achieve the bound in \cref{thm:sampler_with_replacement}, we introduce an alternative way of composing samplers:

\begin{restatable}{proposition}{anotherComposition}
\label{lem:another_matrix_composition}
    Suppose we are given two efficient matrix samplers:
\begin{itemize}
    \item  Let $\Smp_{out}: \zo^{n_1} \times [t_1] \to \zo^m$ be a $(\delta_1, \eps_1)$ matrix sampler.
    \item Let $\Smp_{in}: \zo^{n_2} \times [t_2] \to \zo^{n_1}$ be a $(\delta_2, \eps_2)$ \textbf{averaging} sampler.
\end{itemize}
 Then, for uniformly random sources  $X \sim U_{n_2}$, \[
        \Smp(X):= (\Smp_{out}(\Smp_{in}(X, i), j))_{i \in [t_2], j \in [t_1]}
    \]
    is an efficient $(\delta_2, 2\delta_1 + 2\eps_2 + \eps_1)$ matrix sampler for domain $\zo^m$ with $t_1 \cdot t_2$ samples using $n_2$ random bits. 
\end{restatable}

This essentially says, composing a good $(\eps, \eps)$ matrix sampler $\Smp_{out}$ and a good $(\delta, \eps)$ standard averaging sampler $\Smp_{in}$ would give a  good $(\delta, O(\eps))$ matrix sampler. Although this slightly increases the sample complexity, we can use our sub-sampling technique to reduce it later.

Unlike sub-sampling, in the composition of~\cref{lem:another_matrix_composition}$, \Smp_{in}$ generates multiple random seeds for $\Smp_{out}$, and we query all the samples it produces. This approach effectively reduces the error probability of $\Smp_{out}$ from $\eps$ to $\delta$. The key idea is that only an \(O(\eps)\) fraction of the seeds generated by \(\Smp_{in}\) lead to failure in \(\Smp_{out}\), contributing only a tolerable \(O(\eps)\) additive error in the estimate of \(\E f\).  The reasoning is straightforward: at most an \(\eps\) fraction of all possible seeds for \(\Smp_{out}\) cause failure, and with probability \(1 - \delta\), \(\Smp_{in}\) does not oversample these failure seeds by more than an additional \(\eps\) proportion. As a result, the final proportion of failure seeds remains bounded by \(O(\eps)\).

\begin{remark}
We can also define \emph{strong matrix samplers} as a matrix analog of strong averaging samplers. All results for matrix samplers in this paper would hold for strong matrix samplers as well, with proofs following similar arguments. However, for simplicity, we present our results in the non-strong case only.
\end{remark}

\section{Preliminaries}

\paragraph{Notation.}  We use $[t]$ to represent set $\{1, \dots, t\}$. For integer $m$, $U_m$ is a random variable distributed uniformly over $\zo^m$. For random variables $X$ and $Y$, we use $X \approx_{\eps} Y$ to represent the statistical distance (total variation distance) between $X$ and $Y$ is at most $\eps$, i.e., \[
    \max_{T \subseteq \supp(X)} \abs{\Pr_{x \sim X}[x \in T] - \Pr_{y \sim Y}[y \in T]} \le \eps.
\]

We refer to an algorithm as ``efficient'' if it is polynomial-time computable.
For simplicity, we omit domain sizes for samplers and matrix dimensions when context permits. Unless otherwise specified, statements such as “there exists a \( (\delta, \eps) \) sampler” mean that for all \( 0 < \delta \le \eps < 1 \), there exists a \( (\delta, \eps) \) sampler with the stated properties.
\begin{remark}
    \label{rem:choice_of_parameters}
    The condition $\delta \le \eps$ is very mild and holds in nearly all applications.
    This requirement can be relaxed to \( \delta \le  \eps^\alpha \)  for averaging samplers, and to \( \delta \le d\eps^\alpha \) for matrix samplers, where $\alpha$ is an arbitrarily small positive constant. Such relaxations do not alter the results.

    In the extreme case where $\delta > \eps^\alpha$ for every constant $\alpha > 0$, pairwise independence is already a near-optimal averaging sampler (see~\cref{lem:pairwise_independence}). Specifically, this yields an efficient strong sampler with $O(1 / (\delta \eps^2)) \le O(1 / \eps^{2+\alpha})$ samples, using only $O(m + \log (1 / \eps))$ random bits. 
    Similarly, for matrix samplers under the condition \( \delta > d \eps^\alpha \) for all \( \alpha > 0 \), pairwise independence also achieves near-optimal efficiency with \( O(1 / \eps^{2+\alpha}) \) samples and \( O(m + \log (1 / \eps)) \) random bits.
\end{remark}

\subsection{Extractor-Based Sampler}
As mentioned above, averaging samplers are equivalent to extractors. We will introduce this in detail in \cref{sec:extractors}.
Reingold, Vadhan, and Wigderson used this equivalence to achieve the following:

\begin{theorem}[\protect{\cite[Corollary 7.3]{reingold2000entropy}}, see also \protect{\cite[Theorem 6.1]{oded_survey}}]
\label{thm:optimal_existing_sampler}
For every constant $\alpha > 0$, there exists an efficient $(\delta, \eps)$ averaging sampler over $\zo^m$ with $\poly(1 / \eps, \log(1/\delta))$ samples using $m + (1 + \alpha) \cdot \log_2(1/\delta)$ random bits.
\end{theorem}
For ease of presentation, we often denote an extractor-based averaging sampler by $\Smp_{E} : \zo^n \times \zo^{d} \to \zo^m$, where $\Smp_{E}(X, i)$ is the $i$-th output sample point of the sampler using randomness input $X$. Therefore, the sample complexity of $\Smp_{E}$ is $2^d$.

\subsection{Almost $\ell$-wise Independence}
A sequence $Z_1, \dots, Z_t$ is pairwise independent if the marginal distribution of every pair $(Z_{i_1}, Z_{i_2})$ is uniformly random.
Chor and Goldreich~\cite{CG89} proved that using pairwise independence, we can have a sampler using few random bits but with unsatisfying sample complexity.
\begin{lemma}[\cite{CG89}]
    \label{lem:pairwise_independence}
    For all $\delta, \eps > 0$, there exists an efficient strong $(\delta, \eps)$ averaging sampler for domain $\zo^m$ sampler with $O(1 / (\delta \eps^2))$ samples using $O(m + \log(1/\delta) + \log(1 / \eps))$ random bits.
\end{lemma}

Generalizing pairwise independence, an almost $\ell$-wise independent  sequence is a sequence of random variables such that the marginal distribution of every $\ell$ of them is close to uniform. 

\begin{definition}[\cite{NN_almost_k_wise}]
A sequence of random variables $Z_1, \dots, Z_t \in \zo^m$  is said to be $\gamma$-almost $\ell$-wise independent  if for all subsets $S \subseteq [t]$ such that $|S| \le \ell$, \[
    (Z_i)_{i \in [S]} \approx_{\gamma} U_{m \times |S|}.
\]
\end{definition}

In particular, the pairwise independent sequence mentioned above is a $0$-almost $2$-wise independent  sequence. Naor and Naor proved that such sequences can be randomness-efficiently generated.

\begin{lemma}[\cite{alon92}]
    \label{lem:k_wise_complexity}
    There exists an efficient algorithm that uses 
    $(2 + o(1))(\frac{\ell m}{2} + \log\log t) + 2 \log \frac{1}{\gamma}$ 
    random bits to generate a $\gamma$-almost $\ell$-wise independent  sequence $z_1, \dots, z_t \in \zo^m$.
\end{lemma}

Using standard techniques, we have the following concentration bound for almost $\ell$-wise independent  sequences (see \cref{sec:k-wise-proof} for the proof). Similar bounds for exact $\ell$-wise independent  sequences have been shown in \cite{bellare1994randomness, DodisThesis}. 

\begin{restatable}{lemma}{kWiseConcentration}
    \label{lem:k_wise_concentration}
    Let $Z_1, \dots, Z_t \in \zo^m$ be a sequence of $\gamma$-almost $\ell$-wise independent  variables for an even integer $\ell$. Then for every sequence of functions $f_1, \dots, f_t : \zo^m \to [0, 1]$, \[
        \Prb{\abs{\frac{1}{t}\sum_{i=1}^t (f_i(Z_i) - \E f_i)} \le \eps} \ge 1 - \tuple{\frac{25{\ell}}{\eps^2{t}}}^{\ell/2} - \frac{\gamma}{\eps^\ell}.
    \]
\end{restatable}

\subsection{Composition of Samplers}
The idea of composing samplers has been studied before. More specifically, Goldreich proved the following proposition.
\begin{proposition}[\protect{\cite{oded_survey}}]
\label{prop:concatenation}
Suppose we are given two efficient samplers:
\begin{itemize}
    \item A $(\delta, \eps)$ averaging sampler for domain $\zo^m$ with $t_1$ samples using $n_1$ random bits.
    \item A $(\delta', \eps')$ averaging sampler for domain $\zo^{\log t_1}$ with $t_2$ samples using $n_2$ random bits.
\end{itemize}
Then, there exists an efficient $(\delta + \delta', \eps + \eps')$ averaging sampler for domain $\zo^m$ with $t_2$ samples using $O(n_1 + n_2)$ random bits. 
\end{proposition}

\subsection{Averaging Samplers Imply Matrix Samplers}
When Wigderson and Xiao first introduced matrix samplers, they observed that an averaging sampler also functions as a matrix sampler with weaker parameters, though they did not provide a formal proof. We formalize this observation below:
\begin{restatable}{lemma}{samplerImpliesMatrxSampler}
\label{lem:sampler_implies_matrix_sampler}
    A $(\delta, \eps)$ averaging sampler is a $d$-dimensional $(2d^2 \delta, 2 d \eps)$ matrix sampler.
\end{restatable}
The proof is presented in \cref{sec:WXproof}.

\section{Construction of Averaging Samplers}
Our construction is based on a reduction lemma that constructs a sampler for domain $\zo^m$ based on a sampler for domain $\zo^{O(\log (1 / \eps) + \log\log(1 / \delta))}$. 
We exploit the fact that when composing averaging samplers, the final sample complexity depends on only one of the samplers. Our strategy is:
\begin{itemize}
    \item Apply the extractor sampler in \cref{thm:optimal_existing_sampler} as a $(\delta / 2, \eps / 2)$ sampler over domain $\zo^m$. This uses $m + O(\log (1 / \delta))$ random bits and generates $\poly(1/\eps, \log (1 / \delta))$ samples.
    \item By~\cref{prop:concatenation}, we only need to design a $(\delta / 2, \eps / 2)$ averaging sampler over domain $\zo^{O(\log (1 / \eps) + \log\log(1 / \delta))}$ using $O(\log(1 / \delta))$ random bits. The total sample complexity will be equal to the sample complexity of this sampler. For this sampler, we use almost $\ell$-wise independence.     
\end{itemize}

Following the idea of \cref{prop:concatenation}, 
we first prove that composing samplers maintains the properties of a strong sampler.

\begin{lemma}[Strong Composition]
    \label{lem:new_NZ}
    Suppose we are given two efficient averaging samplers:
\begin{itemize}
    \item  Let $\Smp_{out}: \zo^{n_1} \times [t_1] \to \zo^m$ be a $(\delta, \eps)$ sampler.
    \item Let $\Smp_{in}: \zo^{n_2} \times [t_2] \to \zo^{\log t_1}$ be a strong $(\delta', \eps')$ sampler.
\end{itemize}
 Then, for uniformly random sources $X_{1} \sim U_{n_1}$ and $X_{2} \sim U_{n_2}$, \[
        \Smp(X_1 \circ X_2):= (\Smp_{out}(X_1, \Smp_{in}(X_2, i)))_{i \in [t_2]}
    \]
    is an efficient $(\delta + \delta', \eps + \eps')$ strong averaging sampler for domain $\zo^m$ with $t_2$ samples using $n_1 + n_2$ random bits. 
\end{lemma}

\begin{proof}
    Let $f_1, \dots, f_{t_2} :  \zo^m \to [0, 1]$ be an arbitrary sequence of functions, and define $f_{\text{avg}} := \frac{1}{t_2} \sum_{i=1}^{t_2} f_{i}$.
    Since $\Smp_{out}$ is a $(\delta, \eps)$ averaging sampler and $f_{\text{avg}}$ is bounded in $[0, 1]$, we have \[
        \Prb[X_1 \sim U_{n_1}]{\abs{\E_{Y \sim U_{\log t_1}}{f_{\text{avg}}(\Smp_{out}(X_1, Y))} - \E f_{\text{avg}}} \le {\eps}} \ge 1 - \delta.
    \]
    Equivalently, we can express this as 
    \begin{equation}
        \label{eq:extractor_sampler_application1}
       \Prb[X_1 \sim U_{n_1}]{\abs{\frac{1}{t_2}\sum_{i=1}^{t_2} \E_{Y \sim U_{\log t_1}}{f_i(\Smp_{out}(X_1, Y))} - \frac{1}{t_2}\sum_{i=1}^{t_2} \E f_{i}} \le {\eps}} \ge 1 - \delta.
    \end{equation}
    
    For an arbitrary $x$, view $f_i(\Smp_{out}(x, \cdot))$  as a Boolean function on domain $\zo^{\log t_1}$. Therefore, since $\Smp_{in}(X_2, 1), \dots, \Smp_{in}(X_2, t_2)$ are generated by a strong $(\delta, \eps)$ sampler,
    \begin{equation}
        \label{eq:Smp(k-1)_guarantee1}
        \Pr_{X_2}\left[{\abs{\frac{1}{t_2}\sum_{i=1}^{t_2} \tuple{f_i(\Smp_{out}(x, \Smp_{in}(X_2, i))) - \E_{Y \sim U_{\log t_1}}{f_i(\Smp_{out}(x, Y))}}} \le {\eps'}}\right] \ge 1 - {\delta'}.
    \end{equation}
  
    By the triangle inequality and a union bound over equations~\eqref{eq:extractor_sampler_application1} and \eqref{eq:Smp(k-1)_guarantee1}, we have \[
    \Prb[X_1, X_2]{
       \abs{\frac{1}{t_2}\sum_{i=1}^{t_2}\tuple{ f_i(\Smp_{out}(x, \Smp_{in}(X_2, i)))  - \E f_i}} \le \eps' + \eps} \ge 1 - \delta' - \delta .
    \]
    This proves that the sampler we constructed is a strong $(\delta + \delta', \eps + \eps')$ averaging sampler.
\end{proof}

Instantiating \cref{lem:new_NZ} with the extractor-based sampler from \cref{thm:optimal_existing_sampler} gives:

\begin{lemma}[Reduction Lemma]
\label{lem:one_step_reduction}
    For any $\alpha > 0$: For a sufficiently large constant $C > 0$,  suppose there exists an efficient $({\delta'}, \eps')$ averaging sampler $\Smp_{\text{base}}$ for domain $\zo^{C(\log (1 / \eps) + \log\log(1 / \delta))}$ with $t$ samples using $n$ random bits. Then there exists an efficient $(\delta + \delta', \eps + \eps')$ averaging sampler $\Smp$ for domain $\zo^m$ with $t$ samples using 
        $m + (1 + \alpha)\log (1 / \delta) + n$
        random bits. Moreover, if $\Smp_{base}$ is a strong sampler, then $\Smp$ is also strong.
\end{lemma}

\begin{proof}
    By \cref{thm:optimal_existing_sampler}, there exists an explicit $(\delta, \eps)$ averaging sampler $\Smp_{E}: \zo^{n'} \times \zo^d \to \zo^m$ with \[
        n' = {m + (1 + \alpha)(\log (1 / \delta))} \quad \text{and} \quad  d = \log (\poly(1 / \eps, \log (1 / \delta))) \le  C(\log\log (1 / \delta) + \log (1/\eps))\]
        when $C$ is large enough.
    Therefore, $\Smp_{base}$ can work for domain $\zo^d$.

    This enables us to apply~\cref{lem:new_NZ}. The total number of random bits used is $n' + n = m + (1 + \alpha)\log (1 / \delta) + n$, and a total of $t$ samples are used.



\end{proof}

Next, we show that for domain $\zo^m$ with $m \le O(\log (1 / \eps) + \log\log (1 / \delta))$, we can use an almost $\ell$-wise independent sequence to design a strong averaging sampler with near-optimal sample complexity. 

\begin{lemma}
\label{lem:general_ell_wise_sampler}
    For any $2 \le s < 1 / \delta$, there exists an efficient strong $(\delta, \eps)$ averaging sampler for domain $\zo^m$ with $O(\frac{s}{\eps^{2}}\log \frac{1}{\delta})$ samples using 
    ${\frac{(2 + o(1))(m + 2\log(1 / \eps))\log(1 / \delta)}{\log s}} + 2 \log (1 / \delta)$
    random bits.
\end{lemma}
\begin{proof}
    We begin by setting $\ell = \frac{2\log(2 / \delta)}{\log s}$, $\gamma = \frac{\delta\eps^\ell}{2}$, and $t = \frac{50s \log (2 / \delta)}{  \eps^{2} \log s}$. We then define our sampler by outputting a $\gamma$-almost $\ell$-wise independent  sequence $Z_1, \dots, Z_t \in \zo^m$. 
    Taking the parameters of \cref{lem:k_wise_concentration}, observe \[
        \tuple{\frac{25{\ell}}{\eps^2 {t}}}^{\ell/2}
        = \tuple{\frac{1}{s}}^{\ell/2} = \tuple{\frac{1}{s}}^{\frac{\log(2/\delta)}{\log s}}
        =  \frac{\delta}{2},
    \]
    and \[
        \frac{\gamma}{\eps^\ell} = \frac{\delta}{2}.
    \]
     Therefore, for every sequence of functions $f_1, \dots, f_t : \zo^m \to [0, 1]$, \[
        \Prb{\abs{\frac{1}{t}\sum_{i=1}^t (f_i(Z_i) - \E f_i)} \le \eps} \ge 1 - \delta.
    \]
     Furthermore, \cref{lem:k_wise_complexity} shows that we have an efficient algorithm that uses only \[
     (2 + o(1))\tuple{\frac{\ell m}{2} + \log\log t} + 2 \cdot {\log\frac{1}{\gamma}}  = 
     {\frac{(2 + o(1))(m + 2\log(1 / \eps))\log(1 / \delta)}{\log s}} + 2 \cdot {\log\frac{1}{\delta}}
     \]
    random bits to generate this $\gamma$-almost $\ell$-wise independent  sequence.
\end{proof}

We remark that the construction in \cref{lem:general_ell_wise_sampler} can be replaced with a perfectly \( \ell \)-wise independent sequence. This yields a slightly weaker sampler that uses  
$
O\left( \left( \frac{m + \log(1 / \eps) + \log \log(1 / \delta)}{\log s} + 1 \right) \log \frac{1}{\delta} \right)
$
random bits and  
$
O\left( \frac{s}{\eps^2} \cdot \log \frac{1}{\delta} \right)
$
samples.  
This looser construction is already sufficient to establish \cref{cor:domain_sampler}. However, our tighter construction in \cref{lem:general_ell_wise_sampler} applies more broadly, particularly when \( m \) is small—for example, in \cref{lem:code_sampler}.


Combining \cref{lem:one_step_reduction} and \cref{lem:general_ell_wise_sampler}, we can prove our main result about averaging samplers:

\begin{restatable}{theorem}{generalSampler}
    \label{thm:general_sampler}
     For every constant $\alpha > 0$, and for any $m \ge 1$, $\delta, \eps > 0$, and  $2 \le s \le 1 / \delta$, there exists an efficient strong $(\delta, \eps)$ averaging sampler for domain $\zo^m$  that uses
    \[m + O\tuple{\frac{\log(1/\eps) + \log\log(1 / \delta)}{\log s} \cdot \log \frac{1}{\delta} } + (3 + \alpha) \log \frac{1}{\delta}\] random bits and $O( (s / \eps^2) \cdot \log (1 / \delta))$ samples.  
\end{restatable}

\begin{proof}
 By \cref{lem:one_step_reduction}, our goal is to design an efficient strong $(\delta / 2, \eps / 2)$ averaging sampler $\Smp_{base}$ for domain $\zo^{C(\log(1 / \eps) + \log\log (1 / \delta))}$ for some large enough constant $C$. By~\cref{lem:general_ell_wise_sampler}, for any $2 \le s < 1 / \delta$, such sampler exists using $O(\frac{s}{\eps^{2}}\log \frac{1}{\delta})$ samples and   
 \[
 O\tuple{\frac{(\log (1 / \eps) + \log\log(1 / \delta))\log(1 / \delta)}{\log s}} + 2 \cdot {\log \frac{1}{\delta}}
 \]
 Taking these into \cref{lem:one_step_reduction} gives us the desired bounds.
\end{proof}


For an arbitrarily small constant $\alpha$,  by setting $s = \eps^{-2\alpha}\log^{\alpha}(1 / \delta)$ in~\cref{thm:general_sampler}, we get~\cref{cor:domain_sampler} as a corollary:
\CorollaryTwo*

We can also set $s = 2$ in \cref{thm:general_sampler} and get the following sampler with asymptotically optimal sample complexity but a worse randomness complexity. 

\optimalSampleSampler*



\section{Construction of Matrix Samplers}
Before moving further, we note that non-explicitly, 
a good matrix sampler exists.
This generalizes the non-explicit sampler given in~\cite{CEG95}, with the proof deferred to~\cref{sec:non_explicit}.

\nonExplicit*

Our improved averaging sampler directly implies the best randomness-optimal matrix sampler to date.
Applying \cref{lem:sampler_implies_matrix_sampler} to our sampler in \cref{cor:domain_sampler} gives:
\begin{lemma} 
\label{lem:ours_based_matrix_sampler}
    For every constant $\alpha > 0$, there exists an efficient $d$-dimensional $(\delta, \eps)$ matrix sampler for domain $\zo^m$ using $m + O(\log (d / \delta))$ random bits and $O((\frac{d^2}{\eps^2}\log\frac{1}{\delta})^{1+\alpha})$ samples.
\end{lemma}
However, compared to the optimal sample complexity given in the non-explicit construction, our dependence on $d$ is exponentially worse.
As $d$ is potentially very large, our goal is to utilize our composition to reduce the sample complexity while not increasing the randomness complexity too much.

\subsection{One-Layer Composition}
It is easy to verify that the composition lemma holds for matrices:
\begin{lemma}[Matrix  Composition]
\label{lem:matrix_composition}
    Suppose we are given two efficient matrix samplers:
\begin{itemize}
    \item  Let $\Smp_{out}: \zo^{n_1} \times [t_1] \to \zo^m$ be a $(\delta_1, \eps_1)$ matrix sampler.
    \item Let $\Smp_{in}: \zo^{n_2} \times [t_2] \to \zo^{\log t_1}$ be a $(\delta_2, \eps_2)$ matrix sampler.
\end{itemize}
 Then, for uniformly random sources $X_{1} \sim U_{n_1}$ and $X_{2} \sim U_{n_2}$, \[
        \Smp(X_1 \circ X_2):= (\Smp_{out}(X_1, \Smp_{in}(X_2, i)))_{i \in [t_2]}
    \]
    is an efficient $(\delta_1 + \delta_2, \eps_1 + \eps_2)$ matrix sampler for domain $\zo^m$ with $t_2$ samples using $n_1 + n_2$ random bits. 
\end{lemma}

The proof is essentially the same as the proof of Lemma~\ref{lem:new_NZ}, since the triangle inequality applies to spectral norms, but since we are not dealing with the strong case we only have to do a union bound for two events, a bad sample from $\Smp_{out}$ and a bad sample from $\Smp_{in}$.

The following lemma is the matrix version of \cref{lem:general_ell_wise_sampler}, and we delay its proof to \cref{sec:matrix_lemma_proof}.

\begin{restatable}{lemma}{matrixEllWiseSampler}
\label{lem:matrix_ell_wise_sampler}
    For any $2 \le s < d / \delta$, there exists an efficient $d$-dimensional $(\delta, \eps)$ matrix sampler for domain $\zo^m$ with $O(\frac{s}{\eps^{2}}\log \frac{d}{\delta})$ samples using 
    $O(\frac{(m + \log(1/\eps)) \log (d / \delta)}{\log s} + \log (d / \delta))$
    random bits.    
\end{restatable}
Composing \cref{lem:ours_based_matrix_sampler} with \cref{lem:matrix_ell_wise_sampler} gives us the next theorem.

\begin{lemma}
    \label{lem:matrix_sample_size_reduction}
    Suppose we have an efficient $d$-dimensional $(\delta_1, \eps_1)$ matrix sampler for domain $\zo^m$ with $t$ samples using $n$ bits. For any constant $\alpha > 0$ such that $(t / \eps_2)^\alpha \le d / \delta_2$, we can construct an efficient $d$-dimensional $(\delta_1 + \delta_2, \eps_1 + \eps_2)$ matrix sampler for domain $\zo^m$ with $O(\frac{t^\alpha}{\eps_2^{2+\alpha}} \log \frac{d}{\delta_2} )$ samples using $n + O(\log (d / \delta_2))$ bits.
\end{lemma}
\begin{proof}
    When $(t / \eps_2)^\alpha \le d / \delta$, by setting $s = (t / \eps_2)^\alpha$ in \cref{lem:matrix_ell_wise_sampler}, we have a strong $(\delta_2, \eps_2)$ matrix sampler for domain $\zo^{\log t}$ with $O(\frac{{t}^\alpha}{\eps_2^{2+\alpha}}\log \frac{1}{\delta_2})$ samples using 
    \[
        O\left(\frac{(\log t + \log(1/\eps_2)) \log (d / \delta_2)}{\log s} + \log (d / \delta_2)\right) = O(\log (d / \delta_2))
    \]
    random bits. Then by \cref{lem:matrix_composition}, we have the theorem we want.
\end{proof}

\polydSampler*
\begin{proof}
    Set $\delta_1 = \delta_2 = \delta / 2$ and $\eps_1 = \eps_2 = \eps / 2$ in  \cref{lem:matrix_sample_size_reduction} and apply it to \cref{lem:ours_based_matrix_sampler} will give the result.
\end{proof}

\subsection{Iterated Composition}

\begin{lemma}
    \label{lem:induction_lemma}
    There's an efficient $d$-dimensional $(\delta, \eps)$ matrix sampler for domain $\zo^m$ using $m + O(\log(d / \delta) \log\log d)$ random bits and $O((\frac{\log\log d}{\eps})^5\log^{2}\frac{d}{\delta})$ samples.
\end{lemma}

    \begin{proof}
    Let $r = \log\log d$.
    We will prove that there exists a constant $C > 0$ such that for each $i \in [r]$, there exists an efficient $\left( \frac{i \delta}{r}, \frac{i \varepsilon}{r}\right)$ matrix sampler using at most $m + i \cdot C \log\frac{d}{\delta}$ random bits and $t_i$ samples, where
    \[
        t_i = d^{2^{3-i}} \cdot C \cdot \frac{r^5}{\eps^5} \log^2\frac{d}{\delta}.
    \]
    The $i = r$ case proves the lemma. 

    We will prove by induction on $i$ from $1$ to $r$.
    \paragraph{Base Case ($i = 1$):}
    By \cref{lem:ours_based_matrix_sampler}, there exists an efficient $d$-dimensional $(\frac{\delta}{r}, \frac{\varepsilon}{r})$ matrix sampler using $m + C_1( \log\frac{d}{\delta/r})$ random bits and $C_2\frac{d^3}{(\eps/r)^3}\log^{1.5} \frac{r}{\delta}$ samples for some constants $C_1, C_2 > 0$. When $C \ge 2C_1$ and $C \ge C_2$, we have \[
        m + C_1\tuple{\log\frac{d}{\delta/r}} \le m + C\left( \log\frac{d}{\delta} \right)
    \]
         and
    \[
    C_2\frac{d^3}{(\eps/r)^3}\log^{1.5} \frac{r}{\delta} \le d^{2^{2}} \cdot C \cdot \frac{r^5}{\eps^5} \log^2\frac{d}{\delta} \le t_1.
    \]

    \paragraph{Inductive Step:}
    Assume that for some $i \in [1, r-1]$, there exists an efficient $\left( \frac{i \delta}{r}, \frac{i \varepsilon}{r}\right)$ matrix sampler using $m + i \cdot C\log\frac{d}{\delta}$ random bits and $t_i$ samples.
    By choosing some constant $\alpha < 1/2$ such that $(t_i r / \eps)^\alpha < d r / \delta$ in \cref{lem:matrix_sample_size_reduction}, we have a $\left( \frac{(i+1) \delta}{r}, \frac{(i+1) \varepsilon}{r}\right)$ matrix sampler using $m + i \cdot C\log\frac{d}{\delta} + C_3 \log \frac{d}{\delta/ r}$ random bits and $C_4\frac{\sqrt{t_i}}{(\eps / r)^{2.5}}\log\frac{d}{\delta / r}$ samples for some constants $C_3$ and $C_4$. When $C \ge 2C_3$ and $\sqrt{C} \ge 2C_4$, we have\[
        m + i \cdot C\log\frac{d}{\delta} + C_3 \log \frac{d}{\delta / r} \le m + (i + 1) \cdot C\log\frac{d}{\delta}
    \]
    and \[
        C_4\frac{{t_i}^\alpha}{(\eps / r)^{2+\alpha}}\log\frac{d}{\delta / r} \le 2 C_4 \cdot \sqrt{C} \cdot \sqrt{d^{2^{3-i}}} \cdot \frac{r^5}{\eps^5} \log^2\frac{d}{\delta} \le d^{2^{3-(i+1)}} \cdot C \cdot \frac{r^5}{\eps^5} \log^2\frac{d}{\delta} \le t_{i+1}.
    \]
    This finishes the induction and proves the lemma.
\end{proof}
Using \cref{lem:matrix_sample_size_reduction}, we have the following lemma:
\begin{lemma}
    \label{lem:sampler_without_replacement}
    For any constant $\alpha > 0$:
    There exists an efficient $d$-dimensional $(\delta, \eps)$ matrix sampler for domain $\zo^m$ using $m + O(\log(d / \delta) \log\log d)$ random bits and $O((\frac{1}{\eps^2}\log\frac{d}{\delta})^{1 + \alpha})$ samples.
\end{lemma}
\begin{proof}
    Set $\delta_1 = \delta_2 = \delta / 2$ and $\eps_1 = \eps_2 = \eps / 2$ in  \cref{lem:matrix_sample_size_reduction} and apply it to \cref{lem:induction_lemma} will give the result.
\end{proof}
\subsection{Another Composition Scheme}
To further reduce the number of random bits used in \cref{lem:sampler_without_replacement}, we introduce another way of composing matrix samplers. 
Instead of sub-sampling the samples, we sample the random seeds here.

\anotherComposition*

\begin{proof}
    Let $f : \zo^m \to \C^{d \times d}$ be a function such that $\|f(x)\| \leq 1$ for all $x \in \zo^m$.
    We define $h_f:\zo^{n_1} \to \{0, 1\}$ as follows: \[
        h_f(x) := \mathbf{1}\left[\norm{\frac{1}{t_1}\sum_{i=1}^{t_1} f(\Smp_{out}(x, i)) - \E f} > \eps_1\right].
    \] Then $\E h_f < \delta_1$. One good property of $h_f$ is that \[
        \norm{\frac{1}{t_1}\sum_{i=1}^{t_1} f(\Smp_{out}(x, i)) - \E f} \le 2h_f(x) + \eps_1.
    \]
    For $Z_1, \dots, Z_{t_2}$ the output of $\Smp_{in}$, we have     \[
    \Pr
    \left[ \abs{\frac{1}{t_2}  \sum_i h_f(Z_i) - \E h_f} \leq \varepsilon_2 \right]
    \geq 1 - \delta_2.
    \]
    Therefore, with $1 - \delta_2$ probability, \[
        {\frac{1}{t_2}  \sum_i h_f(Z_i)} < \delta_1 + \eps_2.
    \]
    Then we have 
    \begin{align*}
        \norm{\frac{1}{t_1t_2}\sum_{i=1}^{t_2}\sum_{j=1}^{t_1}f(\Smp_{out}(Z_i, j)) - \E f} &\le \frac{1}{t_2}\sum_{i=1}^{t_2} \norm{\frac{1}{t_1} \sum_{j=1}^{t_1} f(\Smp_{out}(Z_i, j)) - \E f}  \\
        &\le \frac{1}{t_2}\sum_{i=1}^{t_2} (2h_f(Z_i) + \eps_1) \\
        &\le \eps_1 +  {\frac{2}{t_2}  \sum_i h_f(Z_i)}.
    \end{align*}
    This show that, with probability $1 - \delta_2$, the error of $\Smp$ is at most $\eps_1 + 2\delta_1 + 2\eps_2$.
\end{proof}
\samplerWithReplacement*
\begin{proof}
    We apply \cref{lem:another_matrix_composition} with the following choices:
    \begin{itemize}
        \item $\Smp_{out}$: Use the $(\eps / 5, \eps / 5)$ matrix sampler for domain $\zo^m$ in \cref{lem:sampler_without_replacement} by choosing $\alpha = 0.5$. This uses $m + O( \log(d / \eps) \log\log d)$ random bits and $O((\frac{1}{\eps^2}\log\frac{d}{\eps})^{1.5})$ samples.
        \item $\Smp_{in}$: Use the $(\delta, \eps / 5)$ averaging sampler for domain $\zo^{m + O( \log(d / \eps) \log\log d)}$ in \cref{cor:domain_sampler} by choosing $\alpha = 0.5$. This uses $m + O( \log (1 / \delta) + \log(d / \eps) \log\log d)$ random bits and $O((\frac{1}{\eps^2}\log\frac{1}{\delta})^{1.5})$ samples.
    \end{itemize}
    This gives as an efficient $(\delta, \eps)$ matrix sampler using $m + O( \log (1 / \delta) + \log(d / \eps) \log\log d)$ random bits and $O((\frac{1}{\eps^2}\log\frac{d}{\eps})^{1.5} \cdot (\frac{1}{\eps^2}\log\frac{1}{\delta})^{1.5})$ samples.   
    
    Set $\delta_1 = \delta_2 = \delta / 2$ and $\eps_1 = \eps_2 = \eps / 2$ in  \cref{lem:matrix_sample_size_reduction} and apply it to this sampler will reduce the sample complexity to $O((\frac{1}{\eps^2}\log\frac{d}{\delta})^{1 + \alpha})$ for any constant $\alpha > 0$.
\end{proof}
\subsection{Proof of \cref{lem:matrix_ell_wise_sampler}}
\label{sec:matrix_lemma_proof}
\subsubsection{Concentration of Random Matrices}
The goal of this section is to prove \cref{lem:matrix_k_wise_concentration}, an analog of \cref{lem:k_wise_concentration} for random Hermitian matrices. Our approach follows standard techniques from the random matrix literature~\cite{Tomczak1974, CGT12, LT13}.

\begin{lemma}
\label{lem:bound_with_Rad}
Let $\{X_i\}_{i \in [t]}$ be a sequence of independent, mean-zero, self-adjoint random matrices of size $d \times d$ such that $\|X_i\| < 1$ for every $i \in [t]$. Then, for every even interger $q \ge 2$, we have
\[
    \Ex{\Tr \tuple{\tuple{\sum_{i=1}^t \eps_i X_i}^q}} \le d (qt)^{q / 2},
\]
where the sequence $\{\eps_i\}$ consists of independent Rademacher random variables. 
\end{lemma}

\begin{proof}
Note that \[
    \Ex{\Tr \tuple{\tuple{\sum_{i=1}^t \eps_i X_i}^q}} = \sum_{v=(v_1,\dots,v_q)\in[t]^{q}} \Ex{\Tr( \eps_{v_1} X_{v_1}\dots \eps_{v_q} X_{v_q})}
\]
When some index appears an odd number of times in the word
$v=(v_1,\dots,v_q)$, the corresponding term vanishes.
Hence we may restrict the sum to
\[
    \mathcal{T}
    \;=\;
    \bigl\{v\in[t]^{q}:\,
          \text{every index in $v$ appears an even number of times}\bigr\},
\]
and we have \[
\Ex{\Tr \tuple{\tuple{\sum_{i=1}^t \eps_i X_i}^q}} = \sum_{v\in \mathcal{T}} \Ex{\Tr(\eps_{v_1}X_{v_1}\dots \eps_{v_q}X_{v_q})} \le \sum_{v\in \mathcal{T}} d \Ex{\|X_{v_1}\dots X_{v_q}\|} \le \sum_{v\in \mathcal{T}} d,
\]
where the last inequality comes from the fact that $\|X_i\| \le 1$ always holds for every $i$.

%
%
%
%
%


It remains to bound $|\mathcal{T}|$.
Since each index appearing in $v\in\mathcal{T}$ must occur an even number of times, one can group the $q$ positions into exactly $q/2$ unordered pairs so that the two entries in each pair carry the same index.
The number of ways to pair is bounded by $q^{q/2}$.
For each pairing there are $t$ choices of index per pair. Therefore,
\[
    |\mathcal{T}| \le q^{q/2} \cdot t^{q/2} \le (qt)^{q/2} .
\]

Therefore,
\[
    \Ex{\Tr \tuple{\tuple{\sum_{i=1}^t \eps_i X_i}^q}} \le d (qt)^{q / 2} \le
    d |\mathcal{T}|
   \le
    d(qt)^{q/2},
\]
which proves the lemma.
\end{proof}

Using a standard symmetrization trick, we get the following lemma, which is an analog of \cref{prop:MZ_inequality}:
\begin{lemma}\label{lem:trace_moment_no_rad}

Let $\{X_i\}_{i \in [t]}$ be a sequence of independent, mean-zero, self-adjoint random matrices of size $d \times d$ such that $\|X_i\| < 1$ for every $i \in [t]$. Then, for every even interger $q \ge 2$, we have
\[
 \Ex{\Tr \tuple{\tuple{\sum_{i=1}^t X_i}^q}}
    \le 
   d \left(4 q t\right)^{q/2} .
\]
\end{lemma}

\begin{proof}
Let us write
\[
   S  :=  \sum_{i=1}^t X_i,\qquad
   S'  :=  \sum_{i=1}^t X_i',
\]
where $\{X_i'\}_{i=1}^t$ is an independent copy of $\{X_i\}_{i=1}^t$, independent also of the original $X_i$’s.  Since each $X_i$ has mean zero, we have
\[
   \E\left[\Tr\left(S^{q}\right)\right]
    = 
   \E\left[\Tr\left(( S - \E[S'] )^{q}\right)\right]
    \le 
   \E\left[\Tr\left(( S - S' )^{q}\right)\right],
\]
where the last inequality follows from Jensen’s inequality.

Next, observe that
\[
   S - S'  =  \sum_{i=1}^t \left(X_i - X_i'\right).
\]
Introduce a fresh sequence of Rademacher random variables $\{\varepsilon_i\}_{i=1}^t$, independent of everything else.  Then, conditional on $\{X_i,X_i'\}_{i=1}^t$, the random matrix
\[
   \sum_{i=1}^t \varepsilon_i \left(X_i - X_i'\right)
\]
has the same distribution as $ S - S'$.
Consequently,
\[
   \E\left[\Tr\left(( S - S' )^{q}\right)\right]
    = 
   \E_{ X,X'} \E_{ \varepsilon}\left[\Tr\left(\left(\sum_{i=1}^t \varepsilon_i (X_i - X_i')\right)^{q}\right)\right].
\]

Note that for each fixed $(X,X',\varepsilon)$,
\[
   \Tr\left(\left(\sum_{i=1}^t \varepsilon_i (X_i - X_i')\right)^{q}\right)
    \le 
   2^{ q-1}\left(\Tr\left((\sum_{i=1}^t \varepsilon_i X_i)^{q}\right)
                 + \Tr\left((\sum_{i=1}^t \varepsilon_i X_i')^{q}\right)\right).
\]
Taking $\E_{\varepsilon}$, then $\E_{X,X'}$, and using that $\{X_i'\}$ has the same law as $\{X_i\}$ gives
\[
   \E\left[\Tr\left((S - S')^{q}\right)\right]
    = 
   \E_{X,X'} \E_{\varepsilon}\left[\Tr\left((\sum_{i=1}^t \varepsilon_i (X_i - X_i'))^{q}\right)\right]
    \le 
   2^{ q} \E_{X} \E_{\varepsilon}\left[\Tr\left((\sum_{i=1}^t \varepsilon_i X_i)^{q}\right)\right].
\]
In combination with the symmetrization bound 
\(\E[\Tr(S^{q})]\le \E[\Tr((S-S')^{q})]\),
this yields
\[
   \E\left[\Tr(S^{q})\right]
    \le 
   2^{ q} \E\left[\Tr\left((\sum_{i=1}^t \varepsilon_i X_i)^{q}\right)\right].
\]
Taking in the bound in~\cref{lem:bound_with_Rad}, we have 
\[
   \E\left[\Tr\left(\left(\sum_{i=1}^t X_i\right)^{q}\right)\right]
    \le 
   2^{ q} \cdot d (q t)^{ \frac{q}{2}}
    = 
   d \left(2^{2} q t\right)^{ \frac{q}{2}}
    = 
   d \left(4 q t\right)^{ \frac{q}{2}}.
\]
\end{proof}

\begin{lemma}
\label{lem:matrix_k_wise_concentration}
    Let $Z_1, \dots, Z_t \in \zo^m$ be a sequence of $\gamma$-almost $\ell$-wise independent  variables for a positive even  integer $\ell$. Then for any $f: \zo^m \to \C^{d \times d}$ such that for all $x \in \zo^m$, $f(x)$ is Hermitian and $\|f(x)\| \le 1$, we have \[
        \Prb{\norm{\frac{1}{t}\sum_{i=1}^t f(Z_i) - \E f} \le \eps} \ge 1 - 
        d \left( \frac{16\ell}{\eps^2 t} \right)^{\ell/2} - \frac{2^\ell \gamma d}{\eps^\ell}.
\]
\end{lemma}

\begin{proof}
    Let $W_i := (f(Z_i) - \E f) / 2$. 
    We have
    \begin{align*}
    \Prb{\norm{\frac{1}{t}\sum_{i=1}^t f(Z_i) - \E f} > \eps} = \Prb{\norm{\sum_{i=1}^t W_i} > \frac{t\eps}{2}} \le \frac{\Ex{\norm{\sum_{i=1}^t W_i}^\ell}}{(t\eps / 2)^\ell}. 
\end{align*}
Note that each $W_i$ is always a Hermitian matrix, so their sum is always Hermitian and therefore normal. Then we have \[
    \norm{\sum_{i=1}^t W_i}^\ell = \norm{\tuple{\sum_{i=1}^t W_i}^\ell}.
\]
Moreover, since $\ell$ is a positive even integer, $(\sum_{i=1}^t W_i)^\ell$ is positive semidefinite. Therefore, \[
\norm{\tuple{\sum_{i=1}^t W_i}^\ell} \le \Tr \tuple{\tuple{\sum_{i=1}^t W_i}^\ell} =
\Tr\tuple{\sum_{i_1, \dots, i_\ell \in [t]} W_{i_1} W_{i_2} \dots W_{i_\ell}} =
\sum_{i_1, \dots, i_\ell \in [t]} \Tr\tuple{  W_{i_1} W_{i_2} \dots W_{i_\ell}}.
\] 
 Let $W_1', \dots, W_t'$ be a sequence of independent random variables where $W_i' := f_i(U_{\zo^m}) - \E f_i$.

 Since the $W_i$'s are $\gamma$-almost $\ell$-wise independent and $|W_i| \le 1$, we have 
 \begin{align*}
    \Ex{\norm{\sum_{i=1}^t W_i}^\ell} &\le \Ex{\sum_{i_1, \dots, i_\ell \in [t]} \Tr\tuple{  W_{i_1} W_{i_2} \dots W_{i_\ell}}} \\
    &= \sum_{i_1, \dots, i_\ell \in [t]} \Ex{\Tr\tuple{  W_{i_1} W_{i_2} \dots W_{i_\ell}}} \\
    &\le \sum_{i_1, \dots, i_\ell \in [t]} \Ex{\Tr\tuple{  W_{i_1}' W_{i_2}' \dots W_{i_\ell}'}} + \gamma d t^\ell \\
    &= \Ex{\Tr \tuple{\tuple{\sum_{i=1}^t W_i'}^\ell}} + \gamma d t^\ell.
 \end{align*}
Therefore, by~\cref{lem:trace_moment_no_rad}, we have \[
    \Ex{\norm{\sum_{i=1}^t W_i}^\ell} \le \Ex{\Tr \tuple{\tuple{\sum_{i=1}^t W_i'}^\ell}} + \gamma d t^\ell \le 
    d \left(4 \ell t\right)^{ {\ell}/{2}} + \gamma d t^\ell.
\]
Hence, \[
 \Prb{\norm{\sum_{i=1}^t W_i} > \frac{t\eps}{2}}  \le \frac{\Ex{\norm{\sum_{i=1}^t W_i}^\ell}}{(t\eps / 2)^\ell} \le d \left( \frac{16\ell}{\eps^2 t} \right)^{\ell/2} + \frac{2^\ell \gamma d}{\eps^\ell}.
\]
\end{proof}

\subsubsection{Almost $\ell$-wise independence for small domains}
We first prove that the concentration analysis for Hermitian matrices directly implies the general case.
\begin{lemma}    \label{lem:Hermitian_implies_general}
    Let $\Smp: \zo^n \to (\zo^m)^t$ be a function. Suppose for any $f: \zo^m \to \C^{2d\times 2d}$ such that for all $x \in \zo^m$, $f(x)$ is Hermitian and $\|f(x)\| \le 1$,     \[
    \Pr_{(Z_1, \dots, Z_t) \sim \Smp(U_n)}
    \left[ \norm{\frac{1}{t}  \sum_i f(Z_i) - \E f} \leq \varepsilon \right]
    \geq 1 - \delta.
    \]
    Then $\Smp$ is a $d$-dimensional $(\delta, \eps)$ matrix sampler.
\end{lemma}

\begin{proof}
    Let $(Z_1, \dots, Z_t) \sim \Smp(U_n)$.
    We are going to prove that for any function $f : \zo^m \to \C^{d \times d}$ such that $\|f(x)\| \leq 1$ for all $x \in \zo^m$, we have
    \[
    \Pr
    \left[ \norm{\frac{1}{t}  \sum_i f(Z_i) - \E f} \leq \varepsilon \right]
    \geq 1 - \delta.
    \]
    For any matrix $A \in \C^{d \times d}$, its Hermitian dilation $\mathcal{H}(A) \in \C^{2d \times 2d}$ is defined by \[
        \mathcal{H}(A) := \begin{bmatrix}
        0 & A \\
        A^* & 0
        \end{bmatrix}.
    \]
    It is easy to verify that $\norm{A} = \norm{\mathcal{H}(A)}$. Then, for function $g : x \mapsto \mathcal{H}(f(x))$, we have
        \[
    \Pr
    \left[ \norm{\frac{1}{t}  \sum_i g(Z_i) - \E g} \leq \varepsilon \right]
    \geq 1 - \delta.
    \]
    Note that we have \[
    \norm{\frac{1}{t}  \sum_i f(Z_i) - \E f} = \norm{\mathcal{H}\tuple{\frac{1}{t}  \sum_i f(Z_i) - \E f}} = \norm{\frac{1}{t}  \sum_i g(Z_i) - \E g}.
    \]
    Hence, \[
        \Pr
    \left[ \norm{\frac{1}{t}  \sum_i f(Z_i) - \E f} \leq \varepsilon \right]
    \geq 1 - \delta.
    \]
\end{proof}

Now we are ready to prove \cref{lem:matrix_ell_wise_sampler}. 
\matrixEllWiseSampler*

\begin{proof}
    We begin by setting $\ell = \frac{2\log(2d / \delta)}{\log s}$, $\gamma = \frac{\delta\eps^\ell}{2^{\ell + 1}d}$,  and $t = \frac{16\ell s}{\eps^{2}}$. We then define our sampler by outputting a $\gamma$-almost $\ell$-wise independent  sequence $Z_1, \dots, Z_t \in \zo^m$. 
    Taking the parameters of \cref{lem:matrix_k_wise_concentration}, observe \[
    d \left( \frac{16\ell}{\eps^2 t} \right)^{\ell/2}  \le d\tuple{\frac{1}{{s}}}^{\ell / 2} = d\tuple{\frac{1}{{s}}}^{\frac{\log(2d/\delta)}{\log s}}
        =  \frac{\delta}{2},
    \]
    and \[
        \frac{2^\ell\gamma d}{\eps^\ell} = \frac{\delta}{2}.
    \]

     Let $\Smp: \zo^n \to (\zo^m)^t$ be a function. Suppose for any $f: \zo^m \to \C^{2d\times 2d}$ such that for all $x \in \zo^m$, $f(x)$ is Hermitian and $\|f(x)\| \le 1$,     \[
    \Pr_{(Z_1, \dots, Z_t) \sim \Smp(U_n)}
    \left[ \norm{\frac{1}{t}  \sum_i f(Z_i) - \E f} \leq \varepsilon \right]
    \geq 1 - \delta.
    \]
    Then $\Smp$ is a $d$-dimensional $(\delta, \eps)$ matrix sampler by \cref{lem:Hermitian_implies_general}.
     Therefore, by \cref{lem:matrix_k_wise_concentration}, we have \[
        \Prb{\abs{\frac{1}{t}\sum_{i=1}^t (f(Z_i) - \E f)} \le \eps} \ge 1 - \delta.
    \]
    Our sampler uses \[
        t = O\tuple{\frac{\ell s}{\eps^2}} = O\tuple{\frac{s}{\eps^2 \log s}\log\frac{d}{\delta}}
    \]
    samples.
     Furthermore, \cref{lem:k_wise_complexity} shows that we have an efficient algorithm that uses only 
     \begin{align*}
         O(\ell m + \log (1 / \gamma) + \log\log t) &=  O\tuple{\ell m + \ell\log(1 / \eps)  + \log(d / \delta) +  \log\log s} \\
         &=  O\tuple{\frac{(m + \log(1/\eps)) \log (d / \delta)}{\log s} + \log (d / \delta)} 
     \end{align*}
     random bits to generate this $\gamma$-almost $\ell$-wise independent  sequence.
\end{proof}

\begin{remark}
    The initial work by Wigderson and Xiao~\cite{WX05} on matrix samplers focused on Hermitian matrices, where each $f(x)$ was assumed to be Hermitian. Nevertheless, as shown in \cref{lem:Hermitian_implies_general}, any sampler that works for Hermitian matrices can naturally be applied to general matrices as well.
\end{remark}

\section{Applications to Extractors and Codes}
\subsection{Applications to Extractors}

\label{sec:extractors}

Zuckerman showed that averaging samplers are equivalent to randomness extractors \cite{zuckerman1997randomness}. Here we state the only direction that we need.

\begin{lemma}[
\cite{zuckerman1997randomness}]
\label{lem:sampler_implies_extractor}
An efficient strong $(\delta, \eps)$ averaging sampler $\Smp : \zo^n \to (\zo^m)^t$ gives an efficient strong $ (n - \log (1 / \delta) + \log (1 / \eps), 2\eps)$ extractor $\Ext : \zo^n \times \zo^{\log t} \to \zo^m$.
\end{lemma}

Applying \cref{lem:sampler_implies_extractor} on \cref{cor:domain_sampler} gives  \cref{thm:our_extractor}:

\ourExtractor*
\begin{proof}
    By \cref{cor:domain_sampler}, for any positive constant $\alpha > 0$, there exists a constant $\lambda > 1$ such that there exists an efficient strong $(\delta, \eps)$ averaging sampler for domain $\zo^m$ with $O(\frac{1}{\eps^{2 + \alpha}} \log^{1 + \alpha} \frac{1}{\delta})$ samples using $\lambda(m + \log\frac{1}{\delta})$ random bits.

To construct the required strong $(k, \eps)$ extractor for every $n$, we set $\delta$ such that $\log(1 / \delta) = \frac{n}{2\lambda} + \log(1 / \eps)$.
Then, we construct an efficient strong $(\delta, \eps)$ sampler $\Smp$ for domain $\zo^m$ where \[
    m = \frac{n}{\lambda} - \log(1 / \delta) > \frac{n}{2\lambda} - \log(1 / \eps) = \Omega(n) - \log(1 / \eps).
\]
    By the above, $\Smp$ uses $n$ random bits and generates $O(\frac{1}{\eps^{2 + \alpha}} \log^{1 + \alpha} \frac{1}{\delta})$ samples.

By~\cref{lem:sampler_implies_extractor}, $\Smp$ implies an efficient strong $( n - \log(1 / \delta) + \log (1 / \eps), 2\eps)$ extractor $\Ext: \zo^n \times \zo^d \to \zo^m$ with $d \le (1 + \alpha)\log (n-k) + (2 + \alpha) \log (1 / \eps) + O(1)$.
It is only left to verify that $n - \log(1 / \delta) + \log(1 / \eps) \le \beta n$ for some constant $\beta < 1$. We have \[
 n - \log(1 / \delta) + \log (1 / \eps) =  n - \frac{n}{2\lambda}  \le \frac{2\lambda - 1}{2\lambda} n.
\]
This proves the theorem.
\end{proof}

If we would like an extractor with the optimal seed length of $d = \log(n - k) + 2 \log(1 / \eps) + O(1)$, we can have the following extractor using \cref{lem:optimal_sample_sampler}.
\OurOptimalExtractor*

\begin{proof}
    By \cref{lem:optimal_sample_sampler}, there exists a constant $\lambda > 1$ such that there exists an efficient strong $(\delta, \eps)$ averaging sampler for domain $\zo^m$ with $O(\frac{1}{\eps^2} \log \frac{1}{\delta})$ samples using $m + \lambda\log\frac{1}{\delta}(\log\log(1 / \delta) + \log(1 / \eps))$ random bits.

To construct the required strong $(k, \eps)$ extractor for every $n$, we set $\delta$ such that $\log(1 / \delta) = \frac{1}{2\lambda} (\frac{n}{\log n + \log(1 / \eps)}) + \log(1/\eps)$.
Then, we construct an efficient strong $(\delta, \eps)$ sampler $\Smp$ for domain $\zo^m$ where 
\begin{align*}
     m &= n - \lambda\log\frac{1}{\delta}(\log\log(1 / \delta) + \log(1 / \eps)) \\
     &\ge n - \frac{n}{2}\frac{\log\log(1 / \delta) + \log (1 / \eps)}{\log n + \log(1 / \eps)} - \log^2(1 / \eps) -  \log(1 / \eps)\log n \\
     &\ge \Omega(n) - \log^2(1 / \eps).
\end{align*}
   
 By the above, $\Smp$ uses $n$ random bits and generates $O(\frac{1}{\eps^{2}} \log \frac{1}{\delta})$ samples.

By~\cref{lem:sampler_implies_extractor}, $\Smp$ implies an efficient strong $( n - \log(1 / \delta) + \log (1 / \eps), 2\eps)$ extractor $\Ext: \zo^n \times \zo^d \to \zo^m$ with $d = \log (n-k) + 2 \log (1 / \eps) + O(1)$.
It is only left to verify that $ n - \log(1 / \delta) + \log (1 / \eps) \le (1 -  \frac{\beta}{\log n + \log (1 / \eps)})n$ for some constant $\beta < 1$. 
We have \[
 n - \log(1 / \delta) + \log (1 / \eps) = n - \frac{1}{2\lambda} (\frac{n}{\log n + \log(1 / \eps)}) \le (1 -  \frac{1}{2\lambda(\log n + \log (1 / \eps))})n.
\]
This proves the theorem.
\end{proof}

\subsection{Application to List-Decodable Codes}
Error-correcting codes are combinatorial objects that enable messages to be accurately transmitted, even when parts of the data get corrupted. Codes have been extensively studied and have proven to be extremely useful in computer science. Here we focus on the combinatorial property of list-decodability, defined below.
\begin{definition}
    A code ${\textup{\textsf{ECC}}} : \zo^n \rightarrow ({\zo^m})^t$ is $(\rho, L)$ list-decodable if for every received message $r \in (\zo^m)^t$, there are at most $L$ messages $x \in \zo^n$ such that $d_H({\textup{\textsf{ECC}}}(x),r) \le \rho t$, where $d_H$ denotes the Hamming distance. A code is binary if $m = 1$.
\end{definition}

We focus on the binary setting, i.e., $m=1$.

\begin{lemma}[\cite{ExtractorCodes}]
    \label{lem:sampler_implies_codes}
    An efficient strong $(\delta, \eps)$ averaging sampler $\Smp : \zo^n \to \zo^t$ over the binary domain gives an efficient binary code that is $ (\rho = \frac{1}{2} - \eps, \delta 2^n)$ list-decodable with code rate $R = n / t$.
\end{lemma}

To construct our codes, we will use our almost $\ell$-wise independence sampler in \cref{lem:general_ell_wise_sampler} directly.
\begin{lemma}
    \label{lem:code_sampler}
    For all constant $\alpha > 0$,
     there exists an efficient strong $(\delta, \eps)$ averaging sampler for binary domain with $O(\frac{1}{\eps^{2 + \alpha}} \log \frac{1}{\delta})$ samples using $n = C\log(1 / \delta)$ random bits for some constant $C \ge 1$.
\end{lemma}
\begin{proof}
    By setting $s = 1 / \eps^\alpha$ and $m = 1$ in \cref{lem:general_ell_wise_sampler}, we have that whenever $1 / \eps^\alpha \le 1 / \delta$, we have a strong $(\delta, \eps)$ sampler with $O(\frac{1}{\eps^2} \log \frac{1}{\delta})$ samples using $O(\log(1 / \delta))$ random bits. When $1 / \eps^\alpha > 1 / \delta$. Using the pairwise independence sampler in \cref{lem:pairwise_independence} for binary domain will satisfy the condition.
\end{proof}

Applying \cref{lem:sampler_implies_codes} to \cref{lem:code_sampler}
gives \cref{thm:our_code}:
\OurCode*
\begin{proof}
    We use the $(\delta, \eps)$ sampler in \cref{lem:code_sampler}, where we choose $\delta$ such that $n = C\log(1 / \delta)$. 
    Applying \cref{lem:sampler_implies_codes} to this sampler implies \cref{thm:our_code}, where $c(\alpha) = 1 / C$ here.
\end{proof}

\label{sec:codes}

\section{Open Problems}
Our work raises interesting open problems.
\begin{itemize}
    
    \item Comparing to the sampler in \cite{reingold2000entropy} which uses $m + (1 + \alpha)\log(1 / \delta)$ random bits, our averaging sampler requires $m + O(\log (1 / \delta))$ random bits. Can we improve our randomness efficiency while maintaining a good sample complexity?
    \item Is there a way to eliminate the additional $\alpha$ in the sample complexity? For $\eps = 1 / \poly(m)$ and $\delta = \exp(-\poly(m))$, can we design an efficient averaging sampler that is asymptotically optimal in both randomness and sample complexity? 
    \item Can we further improve the randomness complexity of our matrix samplers to fully resolve~\cref{problem2}? 
    \item Is it possible to reduce the list size of the list-decodable codes in~\cref{thm:our_code} to $\poly(n)$ using the structure of the list?
    \item Can we construct randomness-efficient $V$-samplers on other normed spaces $V$? 
\end{itemize}

\section*{Acknowledgements}

We thank Kuan Cheng for introducing us to the matrix sampler problem. We thank Shravas Rao for simplifying and slightly improving~\cref{lem:matrix_k_wise_concentration} (although this didn't improve our final result). We thank Oded Goldreich, Dana Moshkovitz, Amnon Ta-Shma, Salil Vadhan, and anonymous reviewers for helpful comments. 

\printbibliography

@article{ren2001best,
  title={On the best constant in {Marcinkiewicz--Zygmund} inequality},
  author={Ren, Yao-Feng and Liang, Han-Ying},
  journal={Statistics \& probability letters},
  volume={53},
  number={3},
  pages={227--233},
  year={2001},
  publisher={Elsevier}
}

@article{zuckerman1997randomness,
  title={Randomness-optimal oblivious sampling},
  author={Zuckerman, David},
  journal={Random Structures \& Algorithms},
  volume={11},
  number={4},
  pages={345--367},
  year={1997},
  publisher={John Wiley \& Sons, Inc. New York, NY, USA}
}

@article{vadhan2012pseudorandomness,
  title={Pseudorandomness},
  author={Salil P. Vadhan},
  journal={Foundations and Trends{\textregistered} in Theoretical Computer Science},
  volume={7},
  number={1--3},
  pages={1--336},
  year={2012},
  publisher={Now Publishers, Inc.}
}

@inproceedings{bellare1994randomness,
  title={Randomness-efficient oblivious sampling},
  author={Bellare, Mihir and Rompel, John},
  booktitle={Proceedings 35th Annual Symposium on Foundations of Computer Science},
  pages={276--287},
  year={1994},
  organization={IEEE}
}

@article{NN_almost_k_wise,
  title={SMALL-BIAS PROBABILITY SPACES-EFFICIENT CONSTRUCTIONS AND APPLICATIONS},
  author={Naor, Joseph and Naor, Moni},
  journal={SIAM Journal on Computing},
  volume={22},
  number={4},
  pages={838--856},
  year={1993},
  publisher={SIAM (Society for Industrial and Applied Mathematics)}
}

@incollection{oded_survey,
  title={A sample of samplers: A computational perspective on sampling},
  author={Goldreich, Oded},
  booktitle={Studies in Complexity and Cryptography. Miscellanea on the Interplay between Randomness and Computation: In Collaboration with Lidor Avigad, Mihir Bellare, Zvika Brakerski, Shafi Goldwasser, Shai Halevi, Tali Kaufman, Leonid Levin, Noam Nisan, Dana Ron, Madhu Sudan, Luca Trevisan, Salil Vadhan, Avi Wigderson, David Zuckerman},
  pages={302--332},
  year={2011},
  publisher={Springer}
}

@article{Nisan_Zuckerman,
  title={Randomness is linear in space},
  author={Nisan, Noam and Zuckerman, David},
  journal={Journal of Computer and System Sciences},
  volume={52},
  number={1},
  pages={43--52},
  year={1996},
  publisher={Elsevier}
}

@inproceedings{reingold2000entropy,
  title={Entropy waves, the zig-zag graph product, and new constant-degree expanders and extractors},
  author={Reingold, Omer and Vadhan, Salil and Wigderson, Avi},
  booktitle={Proceedings 41st Annual Symposium on Foundations of Computer Science},
  pages={3--13},
  year={2000},
  organization={IEEE}
}

@article{CEG95,
  title={Lower bounds for sampling algorithms for estimating the average},
  author={Canetti, Ran and Even, Guy and Goldreich, Oded},
  journal={Information Processing Letters},
  volume={53},
  number={1},
  pages={17--25},
  year={1995},
  publisher={Elsevier}
}

@article{vadhan2007unified,
  title={The unified theory of pseudorandomness: guest column},
  author={Vadhan, Salil},
  journal={ACM SIGACT News},
  volume={38},
  number={3},
  pages={39--54},
  year={2007},
  publisher={ACM New York, NY, USA}
}

@article{Gil93,
  title={A Chernoff bound for random walks on expander graphs},
  author={Gillman, David},
  journal={SIAM Journal on Computing},
  volume={27},
  number={4},
  pages={1203--1220},
  year={1998},
  publisher={SIAM}
}

@article{CG89,
  title={On the power of two-point based sampling},
  author={Chor, Benny and Goldreich, Oded},
  journal={Journal of Complexity},
  volume={5},
  number={1},
  pages={96--106},
  year={1989},
  publisher={Elsevier}
}

@article{BGG93,
  title={Randomness in interactive proofs},
  author={Bellare, Mihir and Goldreich, Oded and Goldwasser, Shafi},
  journal={Computational Complexity},
  volume={3},
  pages={319--354},
  year={1993},
  publisher={Springer}
}

@ARTICLE{GR08,
  author={Guruswami, Venkatesan and Rudra, Atri},
  journal={IEEE Transactions on Information Theory}, 
  title={Explicit Codes Achieving List Decoding Capacity: Error-Correction With Optimal Redundancy}, 
  year={2008},
  volume={54},
  number={1},
  pages={135-150},
  doi={10.1109/TIT.2007.911222}}

@article{ExtractorCodes,
  title={Extractor codes},
  author={{Ta-Shma}, Amnon and Zuckerman, David},
  journal={IEEE transactions on information theory},
  volume={50},
  number={12},
  pages={3015--3025},
  year={2004}
}

@article{RT00,
  title={Bounds for dispersers, extractors, and depth-two superconcentrators},
  author={Radhakrishnan, Jaikumar and {Ta-Shma}, Amnon},
  journal={SIAM Journal on Discrete Mathematics},
  volume={13},
  number={1},
  pages={2--24},
  year={2000},
  publisher={SIAM}
}

@inproceedings{Raz05,
  title={Extractors with weak random seeds},
  author={Raz, Ran},
  booktitle={Proceedings of the thirty-seventh annual ACM symposium on Theory of computing},
  pages={11--20},
  year={2005}
}

@article{SZ99,
  title={Computing with very weak random sources},
  author={Srinivasan, Aravind and Zuckerman, David},
  journal={SIAM Journal on Computing},
  volume={28},
  number={4},
  pages={1433--1459},
  year={1999},
  publisher={SIAM}
}

@article{HILL,
  title={A pseudorandom generator from any one-way function},
  author={H{\aa}stad, Johan and Impagliazzo, Russell and Levin, Leonid A and Luby, Michael},
  journal={SIAM Journal on Computing},
  volume={28},
  number={4},
  pages={1364--1396},
  year={1999},
  publisher={SIAM}
}

@article{BBR,
  title={Privacy amplification by public discussion},
  author={Bennett, Charles H and Brassard, Gilles and Robert, Jean-Marc},
  journal={SIAM journal on Computing},
  volume={17},
  number={2},
  pages={210--229},
  year={1988},
  publisher={SIAM}
}

@inproceedings{IZ,
  title={How to recycle random bits},
  author={Impagliazzo, Russell and Zuckerman, David},
  booktitle={FOCS},
  volume={30},
  pages={248--253},
  year={1989}
}

@article{Zuc07,
 author = {Zuckerman, David},
 title = {Linear Degree Extractors and the Inapproximability of Max Clique and Chromatic Number},
 year = {2007},
 pages = {103--128},
 doi = {10.4086/toc.2007.v003a006},
 publisher = {Theory of Computing},
 journal = {Theory of Computing},
 volume = {3},
 number = {6},
 URL = {https://theoryofcomputing.org/articles/v003a006},
}

@article{Jus72,
  title={Class of constructive asymptotically good algebraic codes},
  author={Justesen, J{\o}rn},
  journal={IEEE Transactions on information theory},
  volume={18},
  number={5},
  pages={652--656},
  year={1972},
  publisher={IEEE}
}

@article{PCP,
author = {Arora, Sanjeev and Lund, Carsten and Motwani, Rajeev and Sudan, Madhu and Szegedy, Mario},
title = {Proof verification and the hardness of approximation problems},
year = {1998},
issue_date = {May 1998},
publisher = {Association for Computing Machinery},
address = {New York, NY, USA},
volume = {45},
number = {3},
issn = {0004-5411},
url = {https://doi.org/10.1145/278298.278306},
doi = {10.1145/278298.278306},
journal = {J. ACM},
month = {may},
pages = {501–555},
numpages = {55}
}

@book{DodisThesis,
title = "PhD thesis: Exposure-resilient cryptography",
author = "Yevgeniy Dodis",
year = "2000",
month = sep,
language = "English (US)",
publisher = "Massachusetts Institute of Technology",
}

@article{alon92,
  title={Simple constructions of almost k-wise independent random variables},
  author={Alon, Noga and Goldreich, Oded and H{\aa}stad, Johan and Peralta, Ren{\'e}},
  journal={Random Structures \& Algorithms},
  volume={3},
  number={3},
  pages={289--304},
  year={1992},
  publisher={Wiley Online Library}
}

@article{CGT12,
  title={The masked sample covariance estimator: an analysis using matrix concentration inequalities},
  author={Chen, Richard Y and Gittens, Alex and Tropp, Joel A},
  journal={Information and Inference: A Journal of the IMA},
  volume={1},
  number={1},
  pages={2--20},
  year={2012},
  publisher={OUP}
}

@book{LT13,
  title={Probability in Banach Spaces: isoperimetry and processes},
  author={Ledoux, Michel and Talagrand, Michel},
  year={2013},
  publisher={Springer Science \& Business Media}
}

@inproceedings{WX05,
  title={A randomness-efficient sampler for matrix-valued functions and applications},
  author={Wigderson, Avi and Xiao, David},
  booktitle={46th Annual IEEE Symposium on Foundations of Computer Science (FOCS'05)},
  pages={397--406},
  year={2005},
  organization={IEEE}
}

@inproceedings{GLSS18,
  title={A matrix expander chernoff bound},
  author={Garg, Ankit and Lee, Yin Tat and Song, Zhao and Srivastava, Nikhil},
  booktitle={Proceedings of the 50th Annual ACM SIGACT Symposium on Theory of Computing},
  pages={1102--1114},
  year={2018}
}

@article{Rud99,
  title={Random vectors in the isotropic position},
  author={Rudelson, Mark},
  journal={Journal of Functional Analysis},
  volume={164},
  number={1},
  pages={60--72},
  year={1999},
  publisher={Elsevier}
}

@article{AW02,
  title={Strong converse for identification via quantum channels},
  author={Ahlswede, Rudolf and Winter, Andreas},
  journal={IEEE Transactions on Information Theory},
  volume={48},
  number={3},
  pages={569--579},
  year={2002},
  publisher={IEEE}
}

@article{Tro12,
  title={User-friendly tail bounds for sums of random matrices},
  author={Tropp, Joel A},
  journal={Foundations of computational mathematics},
  volume={12},
  pages={389--434},
  year={2012},
  publisher={Springer}
}

@inproceedings{Lu2002,
  author    = {Chin{-}Laung Lu},
  title     = {Hyper-encryption against space-bounded adversaries from on-line strong extractors},
  booktitle = {Advances in Cryptology - {CRYPTO} 2002, 22nd Annual International Cryptology Conference, Santa Barbara, California, USA, August 18-22, 2002, Proceedings},
  year      = {2002},
  pages     = {257--271},
  publisher = {Springer},
}

@inproceedings{Vadhan2003,
  author    = {Salil P. Vadhan},
  title     = {On constructing locally computable extractors and cryptosystems in the bounded storage model},
  booktitle = {Advances in Cryptology - {CRYPTO} 2003, 23rd Annual International Cryptology Conference, Santa Barbara, California, USA, August 17-21, 2003, Proceedings},
  year      = {2003},
  pages     = {61--77},
  publisher = {Springer},
}

@inproceedings{Canetti2000,
  author    = {Ran Canetti and Yevgeniy Dodis and Shai Halevi and Eyal Kushilevitz and Amit Sahai},
  title     = {Exposure-Resilient Functions and All-or-Nothing Transforms},
  booktitle = {Advances in Cryptology - {EUROCRYPT} 2000, International Conference on the Theory and Application of Cryptographic Techniques, Bruges, Belgium, May 14-18, 2000, Proceedings},
  year      = {2000},
  pages     = {453--469},
  publisher = {Springer},
}

@inproceedings{Dodis2002,
  author    = {Yevgeniy Dodis and Joel Spencer},
  title     = {On the (Non)Universality of the One-Time Pad},
  booktitle = {43rd Symposium on Foundations of Computer Science {(FOCS} 2002), 16-19 November 2002, Vancouver, BC, Canada, Proceedings},
  year      = {2002},
  pages     = {376--385},
  publisher = {{IEEE} Computer Society},
}

@inproceedings{Kalai2008,
  author    = {Yael Tauman Kalai and Xin Li and Anup Rao and David Zuckerman},
  title     = {Network Extractor Protocols},
  booktitle = {49th Annual {IEEE} Symposium on Foundations of Computer Science, {FOCS} 2008, October 25-28, 2008, Philadelphia, PA, {USA}},
  year      = {2008},
  pages     = {654--663},
  publisher = {{IEEE} Computer Society},
}

@inproceedings{Kalai2009,
  author    = {Yael Tauman Kalai and Xin Li and Anup Rao},
  title     = {2-Source Extractors under Computational Assumptions and Cryptography with Defective Randomness},
  booktitle = {50th Annual {IEEE} Symposium on Foundations of Computer Science, {FOCS} 2009, October 25-27, 2009, Atlanta, Georgia, {USA}},
  year      = {2009},
  pages     = {617--626},
  publisher = {{IEEE} Computer Society},
}

@inproceedings{Dodis2009,
  author    = {Yevgeniy Dodis and Daniel Wichs},
  title     = {Non-Malleable Extractors and Symmetric Key Cryptography from Weak Secrets},
  booktitle = {Proceedings of the 41st Annual {ACM} Symposium on Theory of Computing, {STOC} 2009, Bethesda, MD, USA, May 31 - June 2, 2009},
  year      = {2009},
  pages     = {601--610},
  publisher = {{ACM}},
}

@article{Tro15,
  title={An introduction to matrix concentration inequalities},
  author={Tropp, Joel A and others},
  journal={Foundations and Trends{\textregistered} in Machine Learning},
  volume={8},
  number={1-2},
  pages={1--230},
  year={2015},
  publisher={Now Publishers, Inc.}
}

@article{Bla19,
  title={Optimal streaming and tracking distinct elements with high probability},
  author={B{\l}asiok, Jaros{\l}aw},
  journal={ACM Transactions on Algorithms (TALG)},
  volume={16},
  number={1},
  pages={1--28},
  year={2019},
  publisher={ACM New York, NY, USA}
}

@InProceedings{Agr19,
  author =	{Agrawal, Rohit},
  title =	{{Samplers and Extractors for Unbounded Functions}},
  booktitle =	{Approximation, Randomization, and Combinatorial Optimization. Algorithms and Techniques (APPROX/RANDOM 2019)},
  pages =	{59:1--59:21},
  series =	{Leibniz International Proceedings in Informatics (LIPIcs)},
  ISBN =	{978-3-95977-125-2},
  ISSN =	{1868-8969},
  year =	{2019},
  volume =	{145},
  editor =	{Achlioptas, Dimitris and V\'{e}gh, L\'{a}szl\'{o} A.},
  publisher =	{Schloss Dagstuhl -- Leibniz-Zentrum f{\"u}r Informatik},
  address =	{Dagstuhl, Germany},
  URL =		{https://drops.dagstuhl.de/entities/document/10.4230/LIPIcs.APPROX-RANDOM.2019.59},
  URN =		{urn:nbn:de:0030-drops-112749},
  doi =		{10.4230/LIPIcs.APPROX-RANDOM.2019.59}
}

@article{Tomczak1974,
author = {Tomczak-Jaegermann, Nicole},
journal = {Studia Mathematica},
language = {eng},
number = {2},
pages = {163-182},
title = {The moduli of smoothness and convexity and the Rademacher averages of the trace classes $S_\{p\}$ ($1\leq p < \infty$)},
url = {http://eudml.org/doc/217886},
volume = {50},
year = {1974},
}

@inproceedings{ILL89,
author = {Impagliazzo, R. and Levin, L. A. and Luby, M.},
title = {Pseudo-random generation from one-way functions},
year = {1989},
isbn = {0897913078},
publisher = {Association for Computing Machinery},
address = {New York, NY, USA},
url = {https://doi.org/10.1145/73007.73009},
doi = {10.1145/73007.73009},
booktitle = {Proceedings of the Twenty-First Annual ACM Symposium on Theory of Computing},
pages = {12–24},
numpages = {13},
location = {Seattle, Washington, USA},
series = {STOC '89}
}

\appendix

\section{Proof of \cref{prop:non_explicit}}
\label{sec:non_explicit}

In this section, we extend the non-explicit averaging sampler construction from~\cite{CEG95} to matrix samplers.

\begin{lemma}
    \label{lem:non_explicit_reduction}
    Let $\Smp$ be a $d$-dimensional $(\delta, \eps)$ matrix sampler for domain $\zo^m$ using $t$ samples. Then there exists a $d$-dimensional $(2\delta, 3\eps)$ matrix sampler for domain $\zo^m$ using $m + 2 \log \frac{1}{\delta} + 2 \log d + \log \log \frac{d}{\eps} + 1$ random bits and $t$ samples.
\end{lemma}
\begin{proof}
Our goal is to construct a $(2\delta, 3\eps)$ matrix sampler $\Smp'$ based on $\Smp$.
\paragraph{Output Approximation via Discretization.} 

We construct a discretization grid \( G \subset \mathbb{C} \) by \[
    G := \{a\Delta + b\Delta i \mid a, b \in \Z \text{ and } \abs{a}^2 + \abs{b}^2 \le \Delta^{-2}\},
\]
where $\Delta = \frac{\eps}{d}$.
For each \( x \in \{0,1\}^m \), define an approximation function \( f' \) that rounds each entry of \( f(x) \) to the nearest point in \( G \), yielding \( f' : \{0,1\}^m \to G^{d \times d} \). Since each entry in \( f(x) \) differs from \( f'(x) \) by at most \( \Delta \), the total approximation error per matrix (in spectral norm) is bounded by \( d\Delta \le \eps\) according to \cref{prop:norm_bound_complex}. Thus, \( f' \) has an average that approximates the average of \( f \) within \( \eps \), and the set of all such approximations \( f' \) forms a finite function class, which we denote \( F \).

\paragraph{Bounding the Size of \( F \).}
Each entry of a matrix in \( G^{d \times d} \) has at most \(1 / \Delta^2\)
possible values. The number of possible matrices is therefore bounded by
\[
\left( \frac{1}{\Delta^2} \right)^{d^2} = \tuple{\frac{d}{\eps}}^{2d^2},
\] 
so the total number of functions in \( F \) is 
\[
\abs{F} \le \left( \tuple{\frac{d}{\eps}}^{2d^2} \right)^{2^m} = 2^{2^{m+1}d^2\log\frac{d}{\eps}}.
\]

\paragraph{Probabilistic Reduction of Random Bits.}
For each function \( f' \in F \), let a random seed be called \textit{bad} if the estimate of \( \Smp \) deviates from the true average of \( f' \) by more than \( \eps \). Since \( \Smp \) is a \( (\delta, \eps) \)-sampler, the fraction of bad random seeds for any \( f' \in F \) is at most \( \delta \). By Hoeffding's inequality, if we select \( k \) random seeds independently at random, the probability that more than \( 2\delta k \) of them are bad is at most $
2e^{-2\delta^2 k}$.
Applying a union bound over all \( f' \in F \), the probability that there exists any \( f' \) with more than \( 2\delta k \) bad seeds is at most \( \abs{F} \cdot 2 e^{-2\delta^2 k} \).

\paragraph{Choosing \( k \) and Applying Probabilistic Method.}
Set 
\[
k = \frac{\ln|F| + \ln 2.01}{2\delta^2} \le \frac{2^{m}d^2\log\frac{d}{\eps} + 1}{\delta^2}
\]
so that \( |F| \cdot 2 e^{-2\delta^2 k} < 1 \). With this choice, there exists a set \( K \) of \( k \) random seeds such that, for all \( f' \in F \), the fraction of bad seeds in \( K \) is at most \( 2\delta \). The number of random bits required to select a sequence \( \rho \in K \) is 
\[
\log k \le m + 2 \log \frac{1}{\delta} + 2 \log d + \log \log \frac{d}{\eps} + 1.
\]

\paragraph{Defining the New Sampler \( \Smp' \).}
We define \( \Smp' \) as follows: Select a random seed \( \rho \in K \), and run \( \Smp(\rho) \) to get samples $Z_1, \dots, Z_t \in \zo^m$. With $1 - 2\delta$ probability, we have 
\begin{align*}
\left\| \frac{1}{t} \sum_{i=1}^t f(Z_i) - \E f \right\| &\leq \left\| \frac{1}{t} \sum_{i=1}^t \left( (f(Z_i) -  f'(Z_i)) + (f'(Z_i)  - \E f') + (\E f' - \E f) \right) \right\| \le 3\eps.
\end{align*}
\( \Smp' \) is then a \( (2\delta, 3\eps) \) matrix sampler, using only 
\[
 m + 2 \log \frac{1}{\delta} + 2 \log d + \log \log \frac{d}{\eps} + 1
\]
random bits.
This completes the proof.
\end{proof}

\begin{theorem}[Matrix Chernoff Bound, see~\cite{Tro15}]
    \label{thm:matrix_chernoff_bound}
    Let $X_1, \dots, X_k$ be independent $d \times d$ complex random matrices. Suppose  $\|X_i\| \leq 1$ for all $i \in [k]$. Then, for any $\varepsilon > 0$, the following inequality holds:
    \[
    \Pr\left( \left\| \frac{1}{t} \sum_{i=1}^t \tuple{ X_i - \mathbb{E}[X_i]} \right\| > \varepsilon \right) \leq 2d \cdot \exp\left( - \frac{3 }{8}  t\eps^2 \right).
    \]
\end{theorem}

Applying matrix chernoff bound, we can prove~\cref{prop:non_explicit}.

\nonExplicit*
\begin{proof}
    By \cref{thm:matrix_chernoff_bound}, taking $t = \frac{24}{\eps^2}\log\frac{4d}{\delta}$ independent samples in $\zo^m$ would give a $d$-dimensional $(\delta / 2, \eps / 3)$ matrix sampler for domain $\zo^m$ using $t$ samples and $tm$ random bits. Applying \cref{lem:non_explicit_reduction}, we get a $d$-dimensional $(\delta, \eps)$ matrix sampler for domain $\zo^m$ using $O(\frac{1}{\eps^2}\log\frac{d}{\delta})$ samples and $m + 2\log\frac{1}{\delta} + 2 \log d + \log \log \frac{d}{\eps}$ random bits.
\end{proof}

\section{Proof of \cref{lem:k_wise_concentration}}
\label{sec:k-wise-proof}
\begin{proposition}[Marcinkiewicz–Zygmund inequality~\cite{ren2001best}]
\label{prop:MZ_inequality}
Let $\{X_i, i \geq 1\}$ be a sequence of independent random variables with $\E X_i = 0$, $\E |X_i|^p < \infty$. Then for $p \ge 2$:
\[
    \E \left|\sum_{i=1}^{n} X_i\right|^p \leq C(p)n^{p/2-1}\sum_{i=1}^{n}\E |X_i|^p,
\]
where $C(p) \leq (3\sqrt{2})^p p^{p/2}$.
\end{proposition}

\kWiseConcentration*
\begin{proof}
    Let $W_i := f_i(Z_i) - \E f_i$. 
    We have
    \begin{align*}
    \Prb{\abs{\sum_{i=1}^t W_i} > t\eps} \le \frac{\Ex{\abs{\sum_{i=1}^t W_i}^\ell}}{(t\eps)^\ell}. 
\end{align*}
 Let $W_1', \dots, W_t'$ be a sequence of independent random variables where $W_i' := f_i(U_{\zo^m}) - \E f_i$.  Since the $W_i$'s are $\gamma$-almost $\ell$-wise independent and $|W_i| \le 1$, we have \[
    \Ex{\abs{\sum_{i=1}^t W_i}^\ell} = \Ex{\tuple{\sum_{i=1}^t W_i}^\ell} \le \Ex{\tuple{\sum_{i=1}^t W_i'}^\ell} + \gamma t^\ell = \Ex{\abs{\sum_{i=1}^t W_i'}^\ell} + \gamma t^\ell.
\]
Since $\E W_i' = 0$ and $|W_i'| \le 1$,
they satisfy the conditions for Marcinkiewicz–Zygmund inequality. We have \begin{align*}
    \Ex{\abs{\sum_{i=1}^t W_i'}^\ell} \leq (3\sqrt{2})^\ell \ell^{\ell/2}t^{\ell/2-1}\sum_{i=1}^{t}\E |W_i'|^\ell \le (5\sqrt{\ell t})^\ell.
\end{align*}
Therefore, \[
\frac{\Ex{\abs{\sum_{i=1}^t W_i}^\ell}}{(t\eps)^\ell} \le \left( \frac{25{\ell}}{\eps^2 {t}} \right)^{\ell/2} + \frac{\gamma}{\eps^\ell}.
\]
\end{proof}

\section{Proof of \cref{lem:sampler_implies_matrix_sampler}}

\label{sec:WXproof}

To prove \cref{lem:sampler_implies_matrix_sampler}, we need the following property of matrix norms:
\begin{proposition}
\label{prop:norm_bound_complex}
Let \( A \in \mathbb{C}^{d \times d} \) and define
\[
r = \max_{i, j} |A_{ij}|.
\]
Then the spectral norm of \( A \) satisfies
\[
r \leq \|A\| \leq d r.
\]
\end{proposition}

\begin{proof}
 Select standard basis vectors \( e_i, e_j \in \mathbb{C}^d \) such that \( |A_{ij}| = r \). Then,
\[
\|A\| \geq \frac{|e_i^* A e_j|}{\|e_i\|_2 \|e_j\|_2} = |A_{ij}| = r.
\]
We also have \[
    \|A\| \le \|A\|_F = \sqrt{\sum_{i = 1}^d \sum_{j = 1}^d A_{ij}^2} \le dr.
\]
\end{proof}

\samplerImpliesMatrxSampler*   

\begin{proof}
    Let $Z_1, \dots, Z_t$ be the sampler's output. We define\[
        A := \frac{1}{t}\sum_{i=1}^t f(Z_i).
    \]  Now we fix some $i,j \in [d]$. For all $x \in \zo^m$, we have $|f(x)_{ij}| \le \norm{f(x)} \le 1$ by \cref{prop:norm_bound_complex}. Then, since $Z_i's$ are the output of a $(\delta, \eps)$ averaging sampler, we have \[
        \Prb{\abs{\mathrm{Re}(A_{ij}) - \mathrm{Re}((\E f)_{ij})} \le \eps} \ge 1 - \delta \quad \text{and} \quad \Prb{\abs{\mathrm{Im}(A_{ij}) - \mathrm{Im}((\E f)_{ij})} \le \eps} \ge 1 - \delta,
    \]
    where $\mathrm{Re}(x)$ and $\mathrm{Im}(x)$ are the functions that extract the real part and imaginary part of $x$ respectively. Take a union bound, we have with $1 - 2\delta$ probability,
    \[
        \abs{A_{ij} - (\E f)_{ij}} \le \abs{\mathrm{Re}(A_{ij} - (\E f)_{ij})} + \abs{\mathrm{Im}(A_{ij} - (\E f)_{ij})} \le 2\eps.
    \]
    By a union bound over all entries, with $1 - 2d^2\delta$ probability, all entries have an additive error bounded by $2\eps$, and this implies that $\|A - \E f\| \le 2d\eps$ by \cref{prop:norm_bound_complex}.
\end{proof}

\end{document}